\documentclass[aip,cha,reprint,twocolumn,showkeys,nofootinbib]{revtex4-1}
\usepackage{graphicx}% Include figure files
\usepackage{bm}% bold math
\usepackage[utf8]{inputenc}
\usepackage[OT1]{fontenc}
\usepackage{amsthm}
\usepackage{amsmath}
\usepackage{amsfonts}
\usepackage{amssymb}
\newtheorem{theorem}{Theorem}
\begin{document}
\title{Mind-to-mind heteroclinic coordination:\\ 
model of sequential episodic memory initiation}
\author{V.~S.~Afraimovich}
\affiliation{Instituto de Investigaci\'on en Comunicaci\'on \'Optica, 
Universidad Aut\'onoma de San Luis Potos\'i,
% Karakorum 1470, Lomas 4a. 78220,
San Luis Potos\'{\i}, M\'exico.}
\author{M.~A.~Zaks} 
\affiliation{Institute of  Physics, Humboldt University of Berlin, 
Berlin, Germany}
\affiliation{Research Institute for Supercomputing, 
Nizhny Novgorod State University, Nizhny Novgorod, Russia}
\author{M.~I.~Rabinovich}
\affiliation{BioCircuits Institute, University of California, 
San Diego, La Jolla, California,~USA}

\begin{abstract}
Retrieval of episodic memory is a dynamical process 
in the large scale brain networks. 
In social groups, the neural patterns, associated to specific events 
directly experienced by single members, 
are encoded, recalled and shared by all participants.
Here we construct and study the dynamical model for the formation and 
maintaining of episodic memory in small ensembles of interacting minds
We prove that  the unconventional dynamical attractor of this process 
-- the nonsmooth heteroclinic torus --
is structurally stable within the  Lotka-Volterra-like sets of equations. 
Dynamics on this torus combines absence of chaos with asymptotic instability
of every separate trajectory; its adequate quantitative characteristics 
are length-related Lyapunov exponents. 
Variation of the coupling strength between the participants results 
in different types of sequential switching between metastable states; we 
interpret them as stages in formation and modification of the episodic memory.
\end{abstract}

\keywords{Episodic memory, neuronal networks, 
winnerless competition, heteroclinic trajectories, chaotic binding}
\maketitle

\onecolumngrid

\noindent\textbf{Our ability to graft images and ideas 
into the minds of other humans 
is crucial for the existence of science, technology and literature.  
Participation of a community member in an event or group of events (episode) 
suffices to implant the memory of that episode into the minds 
of the whole community.
In our daily life we take this ability for granted,  
but only the recent advances of measurement technique have disclosed how e.g., 
a movie spectator encodes and transfers aposteriori to listeners 
the neural patterns associated with viewing specific episodes, 
and how these event-specific patterns are shared among the brains.
In the large-scale networks of the individual brain, 
a retrieval of episodic memory occurs in a way of sequential 
switching between the events, the perpetual ``winnerless competition''.
We propose and investigate the mathematical model
for the formation and maintaining of common memory in interacting minds. 
By combining rigorous proofs with numerical studies we show 
that weak coupling between the participant's minds ensures the existence 
of the so-called ``attracting heteroclinic torus''
in the phase space of the model.
Changing the coupling strength, we observe different types of dynamics 
that correspond to various forms of episodic memory.}

\section{Introduction}
\label{sec:introduction}
\hspace*{0.48\textwidth}\begin{minipage}{8.8cm}\sl
Even across different languages, our brains show similar activity, or become
"aligned" when we hear the same idea or story. This amazing neural mechanism 
allows us to transmit brain patterns, sharing memories and knowledge.\\ 
\hspace*{5.9cm} Uri Hasson (2016).
\end{minipage}

\vspace*{.01cm}

\twocolumngrid

Development of technology and science, as well as the sheer existence of oral 
and written literature owes much to the fact that personal participation in 
an event is not a necessary precondition of keeping that event in one's memory:
humans are able to mentally construct episodes when reading or listening to 
recollections of other humans. A recent study, based on 
analysis of magnetic resonance brain imaging during the performance 
of verbal communication tasks, 
traced how neural patterns associated with viewing specific scenes 
in a movie were encoded, recalled, 
and then transferred to a group of listeners 
who had not seen the movie~\cite{Zadbood_2017}.  
It disclosed that event-specific patterns, observed in the brain default mode
network, were shared across the processes of encoding, recall, 
and construction of the same episodes. 
Such studies uncover intimate correspondences between
episodic memory encoding and construction, and underscore the role of
the common language in the transmission of memory to other brains.

Communication in persistent social groups (families, friends, colleagues, 
$\ldots$) is facilitated by common episodic memories: interpersonal knowledge 
of past, shared by the group members~\cite{Bietti_2010}. 
Distributed within the group, such memories serve as a stem 
around which new layers of shareable information are accumulated. 
Notably, episodic memories are not exact replicas of the lives: 
rather, they are organized summaries of experience, encoded in the form of
sequential groups of events~\cite{Heusser_2017}.
According to recent imaging data, the brain areas responsible
for storage and retrieval of episodic memories include 
hippocampus, striatum and the prefrontal 
cortex~\cite{Sadeh_2011,Collin_2017}.

\subsection*{Low-dimensional mind dynamics}

Below, we present and study a low-dimensional model of mind-to-mind 
episodic memory interaction. 
We emphasize from the beginning that we intend 
not to model the brain itself as a system 
but to create a dynamical model for the activity of this system. 
Our ultimate goal is to describe, understand and make predictions 
of mind dynamics, obtaining, in particular, dynamical models 
of specific classes of such activities as 
cognition, creativity, and autobiographic memory. 

Recent technological progress has allowed the researchers
to observe the brain patterns with resolution and clearness 
that could be previously only dreamed of. The prominent role
is currently played by functional magnetic resonance imaging (fMRI) 
that tracks the changes associated with the blood flow through the brain.
Experimental findings indicate that cognition in the human brain, 
as well as the conscience of certain mammals:\\
a) is closer to determinism than to random processes;\\
b) bears the characteristic features of low-dimensional dynamics, and\\
c) manifests itself in the form of sequential metastable 
spatio-temporal patterns. 

We cite just a few pertinent publications:
\begin{itemize}
\item
In a recent study, Ma and Zhang investigate the temporal organization of 
%characteristic
resting-state functional connectivity (RSFC)  in awake rodents 
and humans. They report: ``We found that transitions between RSFC patterns 
were not random but followed specific sequential orders. 
Transitions between RSFC patterns exhibited high reproducibility 
and were significantly above chance'',
and conclude: ``Spontaneous brain activity is not only
nonrandom spatially, but also nonrandom temporally''~\cite{Ma_Zhang_2018}.
\item
By analyzing local field potentials from the cortices of rats under anesthesia, 
Hudson et al. find out that  ``recovery of consciousness occurs after the brain 
traverses a series of metastable intermediate activity 
configurations''~\cite{Hudson_et_al_2014}. They demonstrate that 
``recovery is confined to a low-dimensional subspace'' and conclude that  
``organization of metastable states, along with dramatic dimensionality 
reduction, significantly simplifies the task of sampling the parameter 
space''~\cite{Hudson_et_al_2014}.
%%%%%%%%%%%
\item Analysis of high temporal resolution human fMRI data from a large 
sample of  unrelated individuals in the study of 
Shine et al~\cite{Shine_Breakspear_2017} 
suggests that the  ``integrative core of brain regions  ... manipulates the 
low-dimensional architecture of the brain across an attractor landscape 
via highly conserved modulatory neurotransmitter systems''; 
reconstruction of state-space trajectories unambiguously confirms
``existence of a low-dimensional, dynamic, integrated component that recurs 
across multiple unique tasks  and demarcates a common cognitive architecture 
within the human brain''. 
The authors of the study summarize: ``Global brain states exist along 
a low dimensional manifold''~\cite{Shine_Breakspear_2017}.
\item In the study \cite{Slapsinskaite_2017} of incremental exhaustive cycling 
performed by the group of physically active adults, the participants 
were instructed to monitor bodily regions with discomfort and pain.
Tracking the evolution of pain-attention during the exercises,
the researchers disclosed  the ``dynamical phenomenon of chunking that the 
biological-cognitive system uses to manage larger sequence of information
into smaller units to facilitate information processing''; they concluded that
``the chunks operate on an heteroclinic cycle of metastable states where each 
metastable state itself is a heteroclinic cycle of basic information items''.
\item Finally, experimental studies on the formation of episodic 
memory~\cite{Baldassano_et_al_2017} show how ``cortical structures generate 
event representations during narrative perception and how these events 
are stored to and retrieved from memory. The data-driven approach allows 
to detect event boundaries as shifts between stable patterns of brain 
activity without relying on stimulus annotations
and reveals a nested hierarchy from short events in sensory regions 
to long events in high order areas (including angular gyrus and posterior 
medial cortex), which represent abstract, multimodal situation models.'' 
Below, we interpret such ``shifts'' as heteroclinic switches 
between metastable patterns.
\end{itemize}
In accordance with this convincing evidence, certain kinds of mind activity 
definitely can be (and already are)  a subject for low-dimensional dynamical 
modeling~\cite{Friston_2000,Tognoli_Kelso_2014,Cocchi_et_al_2017}.

Our modeling approach below is based on the following assumptions,
suggested by experimental data:
\begin{enumerate}
\item Sensory, semantic and emotional information is encoded, memorized, stored
and retrieved by  global brain networks. 
\item During perception, encoded patterns are similar for different humans that
share memory representations for the same real-life events~\cite{Chen_2017}. 
\item In the course of continuous perception, the brain automatically segments 
experience into discrete events~\cite{Baldassano_2017}, 
``the meaningful segments of one's  life, the coherent units 
of one's personal history''~\cite{Beal_Weiss_2013}. 
Segmented information is memorized in the form of abstract
patterns at the high level of hippocampus and cortical areas~\cite{Chen_2017}.
\item Memorized events are segmented into chunks~\cite{chunking_2014}.
Temporally organized chunks form episodes, organized into    
sequences changing with environment~\cite{Cohn_Sheehy_Ranganath_2017}. 
\item Recent studies provide evidence that within events, temporal memory
is related to temporal stability 
of brain memory patterns~\cite{Clewett_Davachi_2017}. 
Accordingly, in the phase space the event patterns should display metastability:
In the retrieval process the chunks compete
and form heteroclinic chains of sequentially switching metastable patterns
\cite{Rabinovich_2008,Rabinovich_2015,Bick_and_Rabinovich_2009}.
\end{enumerate}

These assumptions lead us to the simplified dynamical model of the mutual 
mind-to-mind interaction. We demonstrate that the attractor of the model 
in the case of two interacting subsystems (brains) for a wide range 
of parameters is the unconventional object: the two-dimensional non-smooth 
invariant torus. Peculiarity of dynamics upon it is strict absence of chaos, 
contrasted with instability (in the sense of Lyapunov) of each trajectory.
Remarkably, the proper characteristics of the instability are not 
the conventional Lyapunov exponents (average rates of instability growth 
per time unit), but the average rates of instability growth 
per unit of orbit length in the phase space. 
At larger strength of the coupling between the partners, the torus 
undergoes a breakup, and the resulting dynamical pattern indicates some kind 
of cooperative interaction, akin to synchronization in certain features, 
but different from it in the other ones.

The layout of the paper is as follows. In Sect.\ref{sect:model},
starting from general requirements to characteristics of individual
brain dynamics and to kinds of interactions between the brains,
we delineate the class of considered dynamical systems and 
reduce it to a set of coupled units, each one governed by
Lotka-Volterra-like ordinary differential equations. 
Each subsystem features the non-autonomous episodic memory recall; 
mathematically, we interpret it as the closed heteroclinic chain of episodes 
in the long term memory 
under parametric excitation by sequences that come from the partner subsystems.

The bulk of the paper is focused on the simplest case:
unidirectional mind-to-mind entrainment, ``master-slave'' dynamics. 
In Section~\ref{sect:rigorous} we show that the attractor of this system
is the two-dimensional non-smooth  invariant torus. When 
subsystems are uncoupled, this torus appears as the direct product of two
heteroclinic cycles, and, as we rigorously prove, it persists at least
under sufficiently small coupling strength. 
Every trajectory on the torus is a heteroclinic connection joining two
metastable states of equilibrium. 
Hence,  dynamics on the torus is  absolutely non-chaotic. 
Nevertheless, as numerical experiments in Sect.\ref{sect:numerics}  
force us to believe, each  trajectory in the basin 
of this attractor is Lyapunov unstable.
When, at stronger coupling, the torus breaks up, dynamics
in the slaved subsystem turns into alternation of piecewise constant
segments that follows the switches in the master subsystem.   

\section{The basic model of social cooperation}
\label{sect:model}
Dynamical cell assembly coding belongs to prevailing concepts 
in the context of information processing in the individual human brain.
In global functional brain networks these assemblies form different 
spatio-temporal modes.  
When the minds interact, specific networks are responsible 
for the performance of specific cognitive functions in the partners.  
Since the coding occurs on the population level, dynamics of the modes is
usually low-dimensional~\cite{Ma_Zhang_2018}. Low-dimensionality results from
coherent activity of many elements that form modes, and can be extracted 
from the records by application of e.g., principal 
component  analysis~\cite{Hudson_et_al_2014,Shine_Breakspear_2017}. 
We assume that $N$ different spatio-temporal patterns (brain modes) 
$P_i(\vec{r},t),\;i=1,2,\ldots,N$ 
are characterized by a discrete set of spatial coordinates $\vec{r}$. 
Spatial structure of the patterns is influenced, besides physiological factors, 
by the social environment.  
Noteworthy, $P_i(\vec{r},t)$ may have different sense, 
related to the performance of different  cognitive and behavioral tasks. 

The patterns $P_i(\vec{r},t)$ can be based on several brain subnetworks 
like perceptual, memory, and motor brain circuits, 
therefore their intrinsic dynamics can be quite complex.
In certain cases, temporal and spatial patterns of the modes can be separated:
$P_i(\vec{r},t)=Q_i (\vec{r})\,R_i(t)$ where $Q_i (\vec{r})$ describes 
the spatial organization of the $i$-th mode and $R_i(t)$ characterizes 
its temporal evolution. 
Remarkably, the amplitudes $R_i$ cannot be initiated ``from outside'': 
the mode,  absent at a particular moment of time, 
will be absent for all subsequent times.
Suppose that all $R_i(t)$ obey a kinetic equation up to the second order.
If, as  the result of the inferential process, the spatial structure of
the modes is known, then, after factorization, the basic kinetic model 
for a single brain can be written in the 
generalized Lotka-Volterra form: 
\begin{equation}
\label{LV_solo}
\dot{R}_i=R_i
 \big(\overline{\sigma}_i -R_i-\sum_{j\neq i}^N\overline{\rho}_{i j} R_j\big)
       +\varepsilon\zeta_i(R_i),\,i=1,\ldots,N.
\end{equation}
Here $\overline{\sigma}_i$ denotes the excitation rate of the $i$-th mode, 
$\{\overline{\rho}_{ij}\}$ is the cognitive inhibition matrix 
that characterizes the mutual interaction between the modes, and 
$\varepsilon$ parameterizes the environmental state-dependent
fluctuations $\zeta_i$.

Below we describe the interaction between two social partners; generalization
to larger number of participants is straightforward.
Denote the temporal patterns $R_i(t)$ for partners $X$ and $Y$ by the sets of
functions $x_i(t)$ ($i$=$1,\ldots,N_x$) 
and $y_s(t)$ ($s$=$1,\ldots,N_y$) respectively. 
In general, each mode of $X$ should be enabled to interact 
with every mode of $Y$ and vice versa.  
Then, collective dynamics is governed by the system 
\begin{eqnarray}
\label{LV_two}
\dot{x}_i&=&x_i\left(\sigma_i -x_i-\sum_{j\neq i}^{N_x}\rho_{i j} x_j\,
-q \sum_{s=1}^{N_y}\theta_{i s} y_s\right) +\varepsilon\zeta_i(x_i) \nonumber\\
\\
\dot{y}_k&=&y_k\left(\delta_k -y_k-\sum_{k\neq s}^{N_y}\xi_{k s} y_s\,
-p \sum_{s=1}^{N_x}\eta_{k s} x_s\right)
+\varepsilon\overline{\zeta}_k (y_k) \nonumber
\end{eqnarray}
where $\sigma_i$ and $\delta_k$ are the respective sets of excitation rates 
for the participants $X$ and $Y$,
$\{\rho_{ij}\}$ and $\{\xi_{ks}\}$ are their cognitive
inhibition matrices, the parameters $p$ and $q$ measure
the strength of the social interaction, and the interaction itself 
is prescribed by the matrices $\{\theta_{is}\}$ and $\{\eta_{ks}\}$. 
Finally, $\zeta_i$ and $\tilde{\zeta}_k$ are state-dependent 
fluctuations (noise) that will be specified below.

Formally, the system (\ref{LV_two}) is just the decomposition of 
\eqref{LV_solo}.
However, our setup distinguishes 
between the patterns formed by $X$ and those formed by $Y$.
Accordingly, we expect that the largest elements in the matrices 
$\rho,\,\xi,\,\theta,\,\eta$ are of the order one whereas 
the coupling coefficients $p$ and $q$ stay relatively small. 

\subsection{Configurations of social entrainment}

Complex  dynamics of the system (\ref{LV_two}) in the wide domains 
of parameter values  is able to represent 
the evolution of the sequences of events/episodes. 
It is thereby a convenient model for the analysis of mutual social influence on 
the performance of episodic memory. In different regions of its parameter space,
various attractors can be encountered; for example, many of the 64 stationary
solutions (states of equilibrium) are stable in certain parameter ranges.
We are, however, not interested in stable equilibria or in simple limit cycles:
episodic memory, as it is known from the experiments, 
is neither time-independent, nor strictly periodic. 
Hence we seek in the parameter space the domains where
all states of equilibrium are unstable nodes or saddle points. 
Presence of many invariant hyperplanes in the phase space favors formation 
of structurally (within the frame of Lotka-Volterra-like systems) 
stable heteroclinic connections between the saddle points.
In the context of the memory functioning, such connections enable
an efficient dynamical way of information coding
through robust sequential switching,
based on the winnerless competition principle.
The image of this coding in the phase space
is the stable heteroclinic channel (see Fig.~\ref{fig:het_channel} ).   
\begin{figure}[h]
%\centerline{\includegraphics[width=0.4\textwidth]{het_channel.jpg}}
%\centerline{\includegraphics[width=0.4\textwidth]{fig_1.eps}}
\centerline{\includegraphics[width=0.15\textwidth]{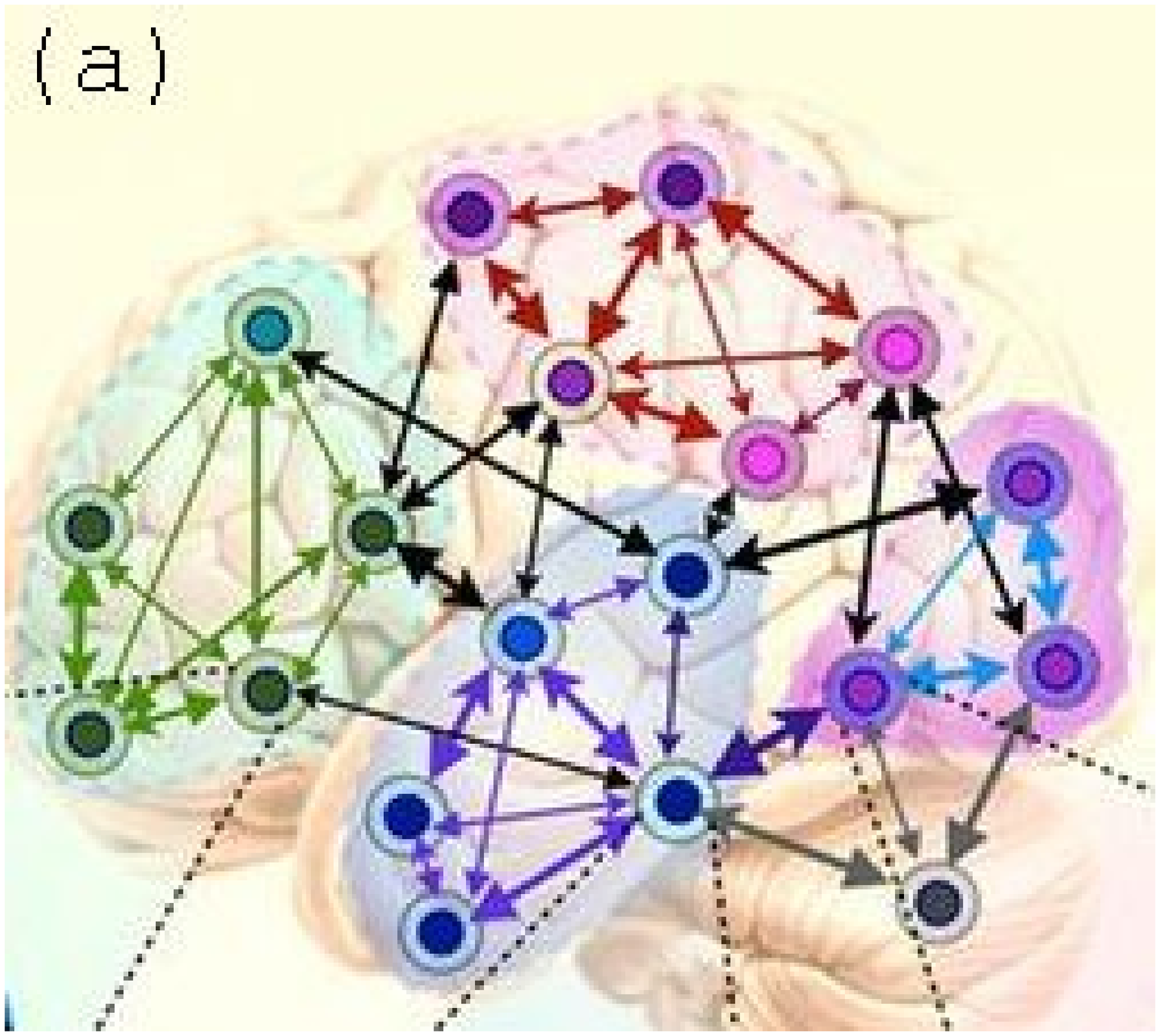}\qquad
\includegraphics[width=0.22\textwidth]{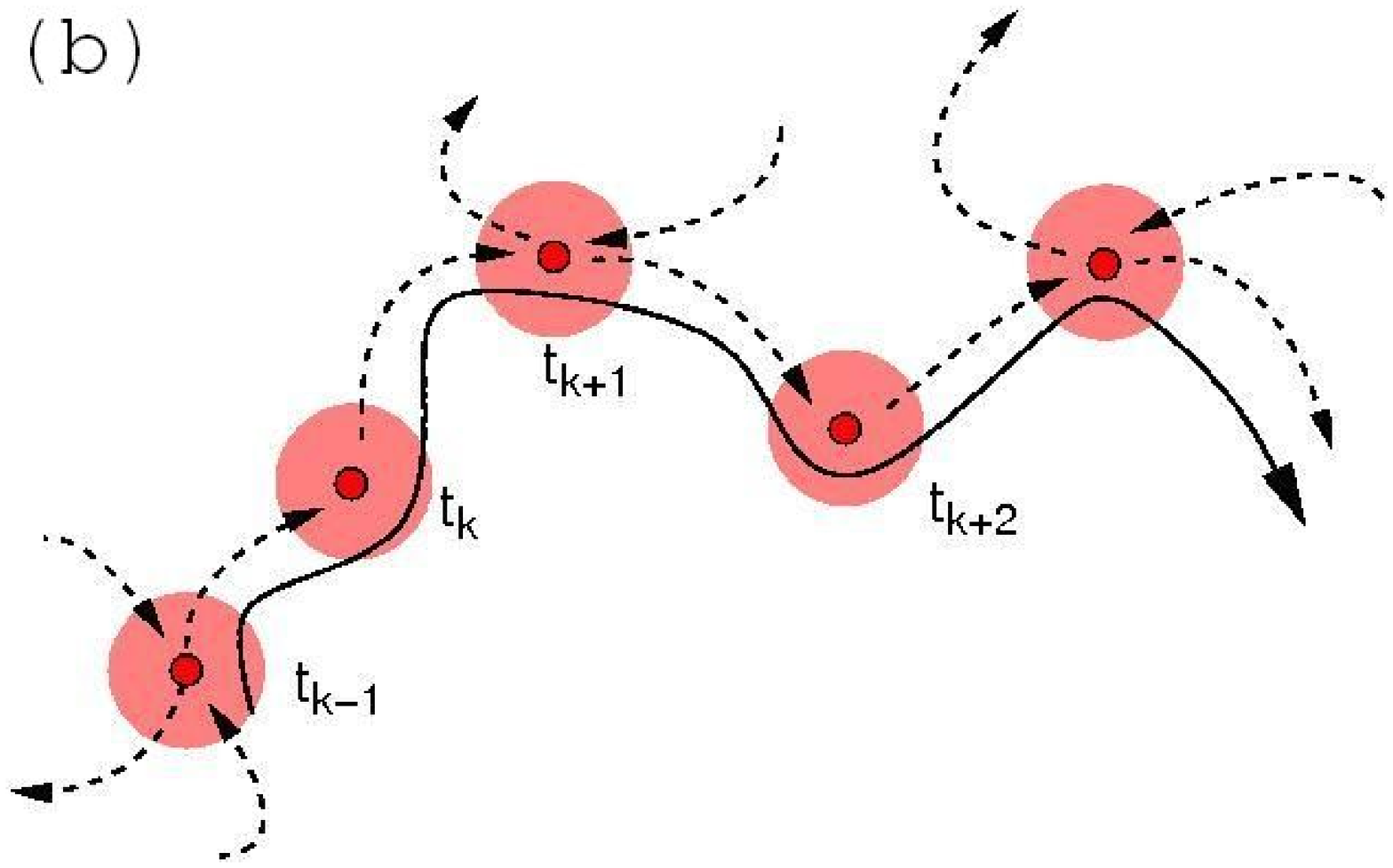}}
\caption{(a) Global brain network whose activity can be represented in the
cognitive space by robust sequential switching; 
excitation level of different network modes is shown by different
colors (b) Stable heteroclinic channel with a chain of metastable states 
(informational patterns). Dashed lines: separatrices of metastable states 
(Adapted from \cite{Varona_Rabinovich_2016}).}
\label{fig:het_channel}
\end{figure}

In our mathematical and numerical studies below, we restrict ourselves 
to the minimal configuration with $N_x$=$N_y$=3: 
operating with just three patterns for each of two partners 
delivers a revealing example of heteroclinic switching 
in the mind-to-mind dynamics. 
Cases of larger $N_x$ and $N_y$, albeit more demanding computationally, can
be treated in the similar way.

If the information exchange between the participants 
is unidirectional ($|q|\ll |p|$) -- this happens, for example, 
if $X$ does not focus her/his attention on the
visual or verbal signals of $Y$ -- the model (\ref{LV_two}) allows
for simplification, enabling the analytic investigation. 

Purely unidirectional connection acquires importance in yet another
situation, relevant for the modeling of memory-related processes. 
Suppose the brain $X$ is not the brain of some other person
but it is the brain of $Y$ in the past. Then the stated problem
turns into the question, how the episodic memory from
the past encodes memory dynamics for  the future: 
a dynamical description of the imagination process. 
Several decades ago, 
D.~Ingvar recognized: in order to be useful, 
a simulation of the future event should be encoded into memory 
so that the gained information can be retrieved at a later time 
when the simulated behavior is actually carried out;
he termed this process ``memory of the future''~\cite{Ingvar_1985} 
(for further details see~\cite{Szpunar_2013,Schacter_2017}). 
From the dynamical theory point of view, ``memory of
the future'' is a result of the inhibitory interaction of the events-modes
from the past episodic memory with modes in the present time 
(see \cite{Varona_Rabinovich_2016}). About the role of memory inhibition 
in imagination of the sequences see  \cite{Gomez-Ariza_2017}. 

This setup differs from heteroclinic harmonic entrainment, 
observed in the single three-state network 
under sinusoidal forcing~\cite{Rabinovich_et_al_2006}:
The localized in time actions of information units
are determined not only by the frequency of heteroclinic cycling
but also by the characteristics (exit times) of metastable states.

\subsection{Dynamical characteristics: length-related Lyapunov exponents.}
When characterizing dynamical regimes in the model (\ref{LV_two}), we
cannot rely on the standard tools like conventional Lyapunov exponents:
in the situation of heteroclinicity to states of equilibrium, they are of 
little help. Recall that the Lyapunov exponents
are defined for the reference trajectory in the system of order $N$ as 
\begin{equation}
\label{lyap_time}
\lambda_i = \lim_{t\to\infty}\frac{1}{t} \log\frac{\| \tilde{x}_i(t) \|}
{\| \tilde{x}_i(0) \|},\;\; i =1, \dots N,
\end{equation}
where $\tilde{x}_i (t)$ are linearly independent solutions of the linearization
near this reference trajectory, that start from $N$ appropriate perturbation 
vectors $\tilde{x}_i (0)$. In our case, we expect that the largest $\lambda_i$ 
vanish and cannot help us to measure the amount of instability 
stored in the attraction basin.

This can be explained by the following reasoning.
A conventional Lyapunov exponent characterizes the rate of perturbation growth
\textit{per unit of time}. 
For trajectories close to heteroclinicity and their perturbations tangent 
to invariant hyperplanes in the phase space of \eqref{LV_two},
the overwhelming  (asymptotically tending to 1) proportion of \textit{time} 
is spent in nearly static configurations, hence a characterization in terms 
of time units loses its merits. 
A more appropriate characteristics of weak instability in this situation 
requires a different parameterization of the trajectory: 
the rate of perturbation growth \textit{per unit of length} 
of the reference trajectory in the phase space,
\begin{equation}
\label{lyap_length}
\Lambda_i = 
\lim_{t\to\infty}\frac{1}{L(t)} \log\frac{\| \tilde{x}_i(t) \|}
{\| \tilde{x}_i(0) \|},\;\; i =1, \dots N,
\end{equation}
where $L(t)$ is the (Euclidean) length of the segment of the reference phase
trajectory between time instants 0 and $t$. This kind of characteristics 
was introduced in ~\cite{AN} where the standard (time-related) Lyapunov exponent
vanished for similar reasons whereas the length-related ones were positive. 
For the  numerical example treated below in Section~\ref{sect:numerics},
in a range of coupling strength there are two positive
length-related exponents $\Lambda_{1,2}$  
and, hence, 
$\|\tilde{x}_{1,2}(t) \| \sim \exp \big(\Lambda_{1,2}L(t)\big)$. 
Recalling that dynamics with two or more positive $\lambda_i$
is termed ``hyperchaos'', here we can speak of weak hyperchaos.

The length-related Lyapunov exponents not only quantify this kind of dynamics
but also serve as indicators %that allow one to identify the most
of essential transitions: in our context, bifurcations of the torus break-up. 
In the exemplary system treated in Section~\ref{sect:numerics},
such bifurcations occur due to the partial regain of stability by equilibria
belonging to the attractor: their formerly two-dimensional unstable manifolds 
become one-dimensional. After this event, only one length-related Lyapunov 
exponent stays positive. Accordingly, dynamics becomes ``one-dimensional'' 
but in a tricky way: time plots of observables are almost piecewise constant, 
with each plateau corresponding to the interval of activity for
one of the master variables (see details below). 

\section{Non-smooth torus: rigorous results}
\label{sect:rigorous}
Here, we introduce and study a new dynamical object: two dimensional
non-smooth invariant torus $\mathbb{T}^ 2$ that can be viewed as a mathematical
image of interactions, in the master-slave way, of two cognitive systems.

When the systems are uncoupled, this object appears as the direct 
product of two heteroclinic cycles, and, as we prove below, 
it persists at least under small rates of coupling. 
Since every trajectory on $\mathbb{T}^ 2$ is a heteroclinic connection between 
two saddle points, dynamics upon it cannot be chaotic. Nevertheless, 
as follows from numerical experiments 
in the subsequent Sect.~\ref{sect:numerics}, 
each trajectory in the basin of this attractor is Lyapunov unstable.
Thus, we deal here with a situation, quite different both from the case
of chaotic attractors and from the phenomenon of transient 
chaos where instability of trajectories is caused by the presence of 
an unstable chaotic set in the boundary of the attractors basin \cite{GOY}. 
For the first time, dynamics of this kind was reported in~\cite{AN}:  
a numerical study of interaction between two systems, 
one of them possessing a heteroclinic cycle and another one  
having a stable limit cycle~\cite{ACY}. 
Two-dimensional sets, entirely consisting of heteroclinic 
connections, were studied also in ~\cite{ARF} and ~\cite{AMY} but 
instability of trajectories in the basin of attractor 
was out of scope of those publications.

In this Section, we treat the variant of  the system (\ref{LV_two}) 
with unilateral coupling and without fluctuating terms:
\begin{eqnarray}
\label{syst_1}
\dot{x}_i& = &x_i \big(\sigma_i - x_i - \sum_{j \neq i} \rho_{ij} x_j \big)\\
\dot{y}_k &=& y_k \big(\delta_k - y_k - \sum_{s \neq k} \xi_{ks} y_s 
        - p \sum_{s = 1}^{3} \eta_{ks} x_s \big)
\label{syst_2}
\end{eqnarray}
where $\sigma_i > 0$, $\rho_{ij} > 0$, 
$\delta_k > 0$, $\xi_{ks} > 0$, $\eta_{ks} \geq 0$, \  $i,j,k,s \in \{1,2,3\}$.

\subsection{The uncoupled system}
Our analysis refers to  small values of the coupling strength $p$. 
We begin with the decoupled case
\begin{eqnarray}
\label{3}
\dot{x}_i &=& x_i (\sigma_i - x_i - \sum_{j \neq i} \rho_{ij} x_j), 
\hspace{.6cm} i,j = 1,2,3\\
\label{4}
\dot{y}_k &=& y_k (\delta_k - y_k - \sum_{s \neq k} \xi_{ks} y_s), 
\hspace{.6cm} k,s = 1,2,3.
\end{eqnarray}
First, we impose conditions under which subsystem \eqref{3} has 
a heteroclinic cycle~\cite{AZR}. 
This system has altogether 8 states of equilibrium.
Of these, three states lie on the coordinate axes. These are  
$O_1 = (\sigma_1,0,0)$, $O_2 = (0, \sigma_2,0)$, $O_3 = (0,0, \sigma_3)$ 
with eigenvalues equal to
\begin{eqnarray}
\sigma_2-\rho_{21} \sigma_1,\; &\sigma_3-\rho_{31}\sigma_1,
& -\sigma_1 \hspace{0.5cm} \text{at} \hspace{0.2cm} O_1,\nonumber\\
\sigma_3-\rho_{32} \sigma_2,\; &\sigma_1-\rho_{12} \sigma_2,\,
&-\sigma_2 \hspace{0.5cm} \text{at} \hspace{0.2cm} O_2,\hspace{0.5cm} \text{and}
\nonumber\\
\sigma_1-\rho_{13} \sigma_3,\;&\sigma_2-\rho_{23} \sigma_3,
&-\sigma_3 \hspace{0.5cm} \text{at} \hspace{0.2cm} O_3.\nonumber
\end{eqnarray}
Under the conditions
\begin{eqnarray}
\label{5}
\sigma_2-\rho_{21} \sigma_1>0, &\;\;& \sigma_3-\rho_{31}\sigma_1 <0,\nonumber\\
\label{6}
\sigma_3-\rho_{32} \sigma_2>0, & \;\;&\sigma_1-\rho_{12}\sigma_2 <0,\\
\label{7}
\sigma_1-\rho_{13} \sigma_3>0, &\;\;&\sigma_2-\rho_{23}\sigma_3 <0.\nonumber
\end{eqnarray}
every $O_i$ ($i$=1,2,3) has the one-dimensional 
unstable and the two-dimensional stable manifolds. 

Moreover, if
$\rho_{21}\rho_{12}\neq 1,\;\rho_{32}\rho_{23}\neq 1$, 
and $\rho_{13}\rho_{31}\neq 1,$
the unstable manifold of $O_i$ contains a heteroclinic trajectory
$\Gamma_{i,(i\,\rm{mod}\, 3)+1}$ joining $O_i$ and $O_{(i\,\rm{mod}\, 3)+1}$. 
Under a combination of all these
conditions, the system \eqref{3} has a heteroclinic cycle 
$$\Gamma = \cup_{i} O_i \cup_{i} \Gamma_{(i\,\rm{mod}\, 3)+1}$$

For the sake of definiteness, we assume that the leading direction on the 
stable manifold of $O_i$ is different from the coordinate axis, i.e.,
\begin{eqnarray}
\label{9}
- \sigma_1 & < & \sigma_3 - \rho_{31} \sigma_1,\nonumber\\ 
- \sigma_2 & < & \sigma_1 - \rho_{12} \sigma_2,\\
- \sigma_3 & < & \sigma_2 - \rho_{23} \sigma_3.\nonumber  
\end{eqnarray}
Then, the heteroclinic cycle has a shape sketched in 
Fig.~\ref{fig:cycle_scheme}a. 

1\begin{figure}[h]
\centerline{\includegraphics[width=0.5\textwidth]{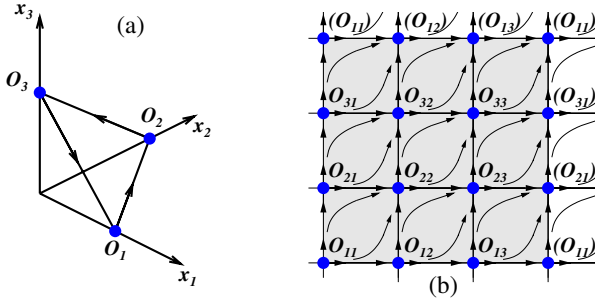}}
\caption{(a) The heteroclinic cycle $\Gamma$;
(b) Unfolding of the torus $T_0$ on the plane: sketch of the vector field.}
\label{fig:cycle_scheme}
\end{figure}
Finally, we impose stability conditions 
\begin{eqnarray}
\label{10}
-\frac{(\sigma_3-\rho_{31}\sigma_1)}{\sigma_2-\rho_{21}\sigma_1}&>& 1,\nonumber\\
-\frac{(\sigma_1-\rho_{32}\sigma_2)}{\sigma_3-\rho_{32}\sigma_2}&>&1,\\ 
-\frac{(\sigma_2-\rho_{23}\sigma_3)}{\sigma_1-\rho_{13}\sigma_3}&>&1,\nonumber
\end{eqnarray}
under which $\Gamma$ is attracting.

Similarly, for the subsystem \eqref{4} we consider three equilibrium points
$\tilde{O}_1 = (\delta_1,0,0)$, $\tilde{O}_2 = (0,\delta_2,0)$, 
$\tilde{O}_3 = (0,0,\delta_3)$ and impose:
\begin{enumerate}
\item[(i).] Conditions for the existence of one-dimensional unstable 
and two-dimensional stable manifolds at $\tilde{O}_i$:
\begin{eqnarray}
\label{11}
\delta_2 - \xi_{21}\delta_1 > 0, &\quad\delta_3 - \xi_{31}\delta_1 < 0,\nonumber\\
\delta_3 - \xi_{32}\delta_2 > 0, &\quad\delta_1 - \xi_{12}\delta_2 < 0,\\ 
\delta_1 - \xi_{13}\delta_3 > 0, &\quad\delta_2 - \xi_{23}\delta_3 < 0.\nonumber
\end{eqnarray}
\item[(ii).] Conditions for leading directions:\\[-5ex]
\begin{eqnarray}
\label{12}
-\delta_1 & <& \delta_3 - \xi_{31}\delta_1,\nonumber\\ 
-\delta_2 & <& \delta_1 -\xi_{12}\delta_2 ,\\ 
-\delta_3 &<&\delta_2 - \xi_{23}\delta_3.\nonumber 
\end{eqnarray}
\item[(iii).] Conditions for the existence of heteroclinic trajectories:
$\xi_{21} \xi_{12} \neq 1, \quad \xi_{32}\xi_{23}\neq 1, 
                         \quad \xi_{13} \xi_{31} \neq 1.$
\item[(iv).] Conditions of stability:
\begin{eqnarray}
\label{14}
-\frac{(\delta_3 - \xi_{31}\delta_1)}{\delta_2 - \xi_{21}\delta_1} > 1,
 &\displaystyle -\frac{(\delta_1 - \xi_{32}\delta_2)}
 {\delta_3 - \xi_{32}\delta_2} > 1,& 
-\frac{(\delta_2 - \xi_{23}\delta_3)}{\delta_1 - \xi_{13}\delta_3} > 1.\nonumber
\end{eqnarray}
\end{enumerate}
Under these provisions, subsystem \eqref{4} has an attractor
\begin{equation*}
\tilde{\Gamma}=\cup_{i = 1}^{3}\tilde{O}_i \cup_{i}
      \tilde{\Gamma}_{(i\,\rm{mod}\, 3)+1},
\end{equation*}
where $\tilde{\Gamma}_{(i\,\rm{mod}\, 3)+1}$ is the heteroclinic trajectory 
joining $\tilde{O}_i$ and $\tilde{O}_{(i\,\rm{mod}\, 3)+1}$. The heteroclinic 
cycle $\tilde{\Gamma}$ looks analogously to $\Gamma$ 
in Fig.~\ref{fig:cycle_scheme}a.

Hence, it follows that the system (\ref{3},\ref{4}) has an invariant set
$T_0$ that is the direct product of $\Gamma$ and $\tilde{\Gamma}$: 
$T_0 = \Gamma \times \tilde{\Gamma}$. Since both $\Gamma$ and $\tilde{\Gamma}$ 
are homeomorphic to the circle, $T_0$ is homeomorphic 
to the  two-dimensional torus $\mathbb{T}^2$. 
The equilibrium points belonging to $T_0$ are 
$O_{ij} = O_i \times \tilde{O}_j$; the eigenvalues of the linearized 
at these points system (\ref{3},\ref{4}) are summarized in the 
Table~\ref{table:eigenvalues}.
\begin{table}[htbp]
\begin{tabular}{|l|c|}
\hline
Saddle & Eigenvalues\\
\hline \hline
$O_{11}$ & $- \sigma_1$, $\sigma_2 - \rho_{21} \sigma_1$, 
   $\sigma_3 - \rho_{31} \sigma_1$, $-\delta_1$, 
   $\delta_2 - \xi_{21}\delta_1$, $\delta_3 - \xi_{31}\delta_1$ \\\hline
$O_{12}$ & $- \sigma_1$, $\sigma_2 - \rho_{21} \sigma_1$, 
   $\sigma_3 - \rho_{31} \sigma_1$, $-\delta_2$, 
   $\delta_3 - \xi_{32}\delta_2$, $\delta_1 - \xi_{12}\delta_2$ \\ \hline
$O_{13}$ & $- \sigma_1$, $\sigma_2 - \rho_{21} \sigma_1$, 
   $\sigma_3 - \rho_{31} \sigma_1$, $-\delta_3$, 
   $\delta_1 - \xi_{13}\delta_3$, $\delta_2 - \xi_{23}\delta_3$ \\ \hline
$O_{21}$ & $- \sigma_2$, $\sigma_3 - \rho_{32} \sigma_2$, 
   $\sigma_1 - \rho_{31} \sigma_1$, $-\delta_1$, 
   $\delta_2 - \xi_{21}\delta_1$, $\delta_3 - \xi_{31}\delta_1$ \\ \hline
$O_{22}$ & $- \sigma_2$, $\sigma_3 - \rho_{32} \sigma_2$, 
   $\sigma_1 - \rho_{12} \sigma_2$, $-\delta_2$, 
   $\delta_3 - \xi_{32}\delta_2$, $\delta_1 - \xi_{12}\delta_2$ \\ \hline
$O_{23}$ & $- \sigma_2$, $\sigma_3 - \rho_{32} \sigma_2$, 
   $\sigma_1 - \rho_{12} \sigma_2$, $-\delta_3$, 
   $\delta_1 - \xi_{13}\delta_3$, $\delta_2 - \xi_{23}\delta_3$ \\ \hline
$O_{31}$ & $- \sigma_3$, $\sigma_1 - \rho_{13} \sigma_3$, 
   $\sigma_1 - \rho_{23} \sigma_3$, $-\delta_1$, 
   $\delta_2 - \xi_{21}\delta_1$, $\delta_3 - \xi_{31}\delta_1$ \\ \hline
$O_{32}$ & $- \sigma_3$, $\sigma_1 - \rho_{13} \sigma_3$, 
   $\sigma_2 - \rho_{23} \sigma_3$, $-\delta_2$, 
   $\delta_3 - \xi_{32}\delta_2$, $\delta_1 - \xi_{12}\delta_2$ \\ \hline
$O_{33}$ & $- \sigma_3$, $\sigma_1 - \rho_{13} \sigma_3$, 
   $\sigma_2 - \rho_{23} \sigma_3$, $-\delta_3$, 
   $\delta_1 - \xi_{13}\delta_3$, $\delta_2 - \xi_{23}\delta_3$ \\ \hline
\end{tabular}
\caption{Eigenvalues of the linearized system  (\ref{3},\ref{4}).}
\label{table:eigenvalues}
\end{table}

The assumed conditions imply that each of the points $O_{ij}$ has
the two-dimensional unstable manifold and the four-dimensional stable manifold
(below we denote these manifolds by, respectively, $W^u$ and $W^s$).
Moreover,  some of $O_{ij}$ are joined by heteroclinic trajectories.
To list them we introduce the following notation: let $H(A \rightarrow B)$ be a
heteroclinic trajectory joining the equilibrium points A and B. For heteroclinic
cycles $\Gamma$ and $\tilde{\Gamma}$ we have the heteroclinic
trajectories: $H(O_1 \rightarrow O_2) =: \Gamma_{12}$, $H(O_2 \rightarrow O_3)
=: \Gamma_{23}$, $H(O_3 \rightarrow O_1) =: \Gamma_{31}$, $H(\tilde{O}_1
\rightarrow \tilde{O}_2) =: \tilde{\Gamma}_{12}$, 
$H(\tilde{O}_2 \rightarrow \tilde{O}_3) =: \tilde{\Gamma}_{23}$, 
and $H(\tilde{O}_3 \rightarrow \tilde{O}_1) =:\tilde{\Gamma}_{31}$. 
Heteroclinic trajectories can be listed in the way presented 
in the Table~\ref{table:heteroclinics}.

%\onecolumngrid
%\begin{widetext} 
\begin{table}[htbp]
\tiny
\begin{tabular}{|l|l|l|l|l|}
\hline \hline
$O_1$ & $H(O_{11}\to O_{12}) = O_1\times\tilde{\Gamma}_{12}$ & 
      $H(O_{12}\to O_{13}) = O_1\times\tilde{\Gamma}_{23}$ & 
      $H(O_{13}\to O_{11}) = O_1\times\tilde{\Gamma}_{31}$\\ \hline
$O_2$ & $H(O_{21}\to O_{22}) = O_2\times\tilde{\Gamma}_{12}$ & 
      $H(O_{22}\to O_{23}) = O_2\times\tilde{\Gamma}_{23}$ & 
      $H(O_{23}\to O_{21}) = O_2\times\tilde{\Gamma}_{31}$\\ \hline
$O_3$ & $H(O_{31}\to O_{32}) = O_3\times\tilde{\Gamma}_{12}$ & 
      $H(O_{32}\to O_{33}) = O_3\times\tilde{\Gamma}_{23}$ & 
      $H(O_{33}\to O_{31}) = O_3\times\tilde{\Gamma}_{31}$\\ \hline
$\tilde{O}_1$ & $H(O_{11}\to O_{21}) = \Gamma_{12}\times\tilde{O}_1$
    & $H(O_{21}\to O_{31}) = \Gamma_{23}\times\tilde{O}_1$ 
    & $H(O_{31}\to O_{11}) = \Gamma_{31}\times\tilde{O}_1$\\ \hline
$\tilde{O}_2$ & $H(O_{12}\to O_{22}) = \Gamma_{12}\times\tilde{O}_2$ 
    & $H(O_{22}\to O_{32}) = \Gamma_{23}\times\tilde{O}_2$ 
    & $H(O_{32}\to O_{12}) = \Gamma_{31}\times\tilde{O}_2$\\ \hline
$\tilde{O}_3$ & $H(O_{13}\to O_{23}) = \Gamma_{12}\times\tilde{O}_3$ 
    & $H(O_{23}\to O_{33}) = \Gamma_{23}\times\tilde{O}_3$ 
    & $H(O_{33}\to O_{13}) = \Gamma_{31}\times\tilde{O}_3$\\ \hline
\end{tabular}
\caption{Heteroclinic trajectories of basic heteroclinic network of $T_0$.}
\label{table:heteroclinics}
\end{table}
%\end{widetext}
%\twocolumngrid
 
For convenience, we place in the left column of this Table the
equilibrium points that enter the corresponding direct products. These 18
heteroclinic trajectories form a basic heteroclinic network: 
see Fig.~\ref{fig:torus_scheme}.
\begin{figure}[h]
\centerline{\includegraphics[width=0.35\textwidth]{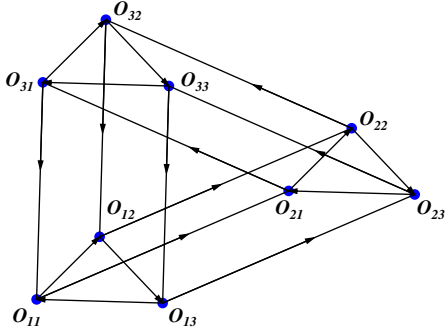}}
\caption{Heteroclinic network $\Gamma_0$ on the torus $T_0$.}
\label{fig:torus_scheme}
\end{figure}
Let us now construct a two-dimensional invariant surface for which $\Gamma_0$
plays the role of "skeleton". For that, we show that each "rectangle" in
$\Gamma_0$ serves as a boundary of a two-dimensional invariant surface. 
Consider, e.g., the rectangle formed by the heteroclinic
trajectories $H(O_{13} \rightarrow O_{11})$, $H(O_{31} \rightarrow O_{11})$,
$H(O_{33} \rightarrow O_{31})$, $H(O_{33} \rightarrow O_{13})$ and the points
$O_{33}$, $O_{31}$, $O_{13}$, $O_{11}$. All
these heteroclinic trajectories and saddle points belong to the
four-dimensional invariant plane $x_2 = y_2 = O$ which we denote by
$\mathbb{R}^{4}$.

Recalling that $H(O_{13} \rightarrow O_{11}) = O_1 \times \tilde{\Gamma}_{31}$,
$H(O_{33} \rightarrow O_{31}) = O_3 \times \tilde{\Gamma}_{31}$ and $H(O_{31}
\rightarrow O_{11}) = \Gamma_{31} \times \tilde{O}_1$, $H(O_{33} \rightarrow
O_{13}) = \Gamma_{31} \times \tilde{O}_3$, we naturally consider the
two-dimensional surface $\Gamma_{31} \times \tilde{\Gamma}_{31} = \tilde{R}_0$
that possesses the following properties:
\begin{enumerate}
\item[(i).] It is invariant by definition.
\item[(ii).] $\tilde{R}_0 \subset \mathbb{R}^{4}$, by definition.
\item[(iii).] $\tilde{R}_0 \subset {W^u}_{O_{33}}$. Indeed, a point in
$\tilde{R}_0$ is the product of two points, say P $\in \Gamma_{31}$ and $Q \in
\tilde{\Gamma}_{31}$. As time goes to $-\infty$, the representative point on
$\Gamma_{31}$ tends to $O_3$ and that on $\tilde{\Gamma}_3$ tends to
$\tilde{O}_3$, so the representative point the trajectory of the full system
going through $P \times Q$ tends to $O_{33}$. 

\item[(iv).] $\tilde{R}_0 \subset {W^s}_{O_{11}}$. 
The proof is the same as in (iii).  Thereby, 
$\tilde{R}_0$ is a collection of heteroclinic connections
joining $O_{33}$ and $O_{11}$, see Fig~\ref{fig:rectangle}.
\end{enumerate}
\begin{figure}[h]
\centerline{\includegraphics[width=0.2\textwidth]{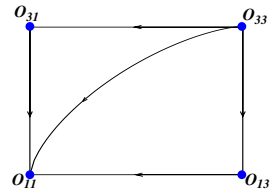}}
\caption{The rectangle $\tilde{R}_0$.}
\label{fig:rectangle}
\end{figure}
In the same way we prove that there are other 8 rectangles with boundaries
consisting of heteroclinic trajectories from the basic network, 
see Fig~\ref{fig:torus_scheme}.
They form a two-dimensional surface, say, $T_0$, homeomorphic to the 
two-dimensional torus. In fact, $T_0$ is the direct product 
of heteroclinic cycles in the systems \eqref{3} and \eqref{4}. 

By construction, the complete set of trajectories on $T_0$ consists of 
nine metastable equilibria $O_{ij}$ and heteroclinic connections 
joining them pairwise.
Since in the full phase space of system (\ref{3},\ref{4}) all $O_{ij}$
are structurally stable saddles, their presence on the toroidal surface 
(in absence of compensating nodes or foci) may seem to violate 
the index theorem. However, reduction to $T_0$
involves ``folding'' along the separatrices of the saddles. As seen in
the schematic unfolding of $T_0$ on the two-dimensional plane 
sketched in Fig.~\ref{fig:cycle_scheme}b, the procedure of folding
turns steady states into compound equlibria: 
in the adjacent segment of the plane, each of them is a source in one quadrant, 
a sink in another one and a saddle point in two remaining quadrants, so that 
the resulting Poincar{\'e} index of every steady state is identically zero.
Accordingly, the total index of $T_0$ vanishes as well, as required
for every two-dimensional toroidal surface.

Now we formulate the sufficient conditions under which $T_0$ is an attractor.
For that, we remind the notion of the saddle index~\cite{Sh}.
If $O$ is a  hyperbolic saddle point with Jacobian
eigenvalues $\lambda_1,\dots, \lambda_m$, $\gamma_1, \dots, \gamma_n$ 
such that $\operatorname{Re}{\lambda_i} < 0$, 
$\operatorname{Re}{\gamma_j} >0$, then the number
\begin{equation*}
\nu = -\frac{\max_{i} \operatorname{Re} \lambda_i}{\max_{j} 
      \operatorname{Re} \gamma_j}
\end{equation*}
is called the saddle index of $O$. If $\nu > 1$, the point $O$ is called the
dissipative saddle. For example, for the point $O_{33}$ the saddle index equals
\begin{equation*}
\nu = - \frac{\max \{\sigma_2 - \rho_{23} \sigma_3,\delta_2 - \xi_{23}
\delta_3\}}{\max \{\sigma_1 - \rho_{13} \sigma_3,\delta_1 - \xi_{13}
\delta_3\}}
\end{equation*}
See Table~\ref{table:eigenvalues} and the inequalities \eqref{5} - \eqref{12}.

\begin{theorem}
\label{theorem_1}
If all saddle points in $\Gamma_0$ are dissipative, then $T_0$ is an attractor.
\end{theorem}

\begin{proof}
It suffices to show that the representative point on the trajectory passing
through any initial point in a neighborhood of $T_0$ tends to $T_0$ as $t
\rightarrow \infty$. Without a loss of generality we start with an initial
point $q_0$ that is close to one of the equilibria in $\Gamma_0$; 
let this equilibrium be, e.g., $O_{33}$. Let
$\epsilon = {\rm dist}(q_0, {W_{O}^{s}}_{33})$. It follows 
(see~\cite{Sh,AH}) that the orbit passing through $q_0$ 
leaves a neighborhood of $O_{33}$ at a point  $q_1$ such that 
${\rm dist}(q_1, {W_{O}^{u}}_{33}) < \epsilon^{\nu_0}$ where 
$1 < \nu_0 < \nu_{33}$ ($\nu_{33}$ denotes the saddle index of $O_{33}$). 
Then the orbit follows a trajectory on
${W_{O}^{u}}_{33}$ and comes, after a finite time, to a point $q_2$
in a small neighborhood of the equilibrium $\tilde{O}$ that is
either $O_{33}$ or $O_{13}$ or $O_{11}$, so that 
${\rm dist}(q_3, W^{s}(\tilde{O})) < C \epsilon^{\nu_0}$, $C = const$. 
If $\epsilon$ has been
small enough, then $C \epsilon^{\nu_0} < \epsilon/2$. 
Then we reproduce the previous consideration for the point $\tilde{O}$, 
replacing $\epsilon$ by $\epsilon/2$. 
Repeating this procedure again and again, 
we ensure that ${\rm dist}(q_{2k+1}, T_0) < {\epsilon}/{2^k} $ 
where $q_{2k+1}$ is
the representative point after the time in which the trajectory intersects $k$
successive neighborhoods. Remark that we should choose constant $C$ only
finitely many times independently of $k$ since the passage time from one
neighborhood of an equilibrium  to another one in the same rectangle 
is bounded from above. Thus, the representative point tends to $T_0$ 
as $t \rightarrow +\infty$.
\end{proof}
Table~\ref{table:indices} shows the saddle indices of
all saddles in $\Gamma_0$.
\begin{table}[htbp]
\begin{center}
\begin{tabular}{|l|c|}
\hline
Saddle & Saddle index $\nu$\\
\hline \hline
$O_{11}$ & $ \frac{- \max \{\sigma_3 - \rho_{31} \sigma_1,\delta_3 - \xi_{31}\delta_1 \}}{ \max \{\sigma_2 - \rho_{21} \sigma_1,\delta_2 - \xi_{21}\delta_1\}}$  \\\hline
$O_{12}$ & $ \frac{-\max \{\sigma_3 - \rho_{31} \sigma_1,\delta_1 - \xi_{12}\delta_2\}}{\max \{\sigma_2 - \rho_{21} \sigma_1,\delta_3 - \xi_{32}\delta_2\}}$ \\\hline
$O_{13}$ &  $\frac{-\max \{\sigma_3 - \rho_{31} \sigma_1,\delta_2 - \xi_{23}\delta_3\}}{\max \{\sigma_2 - \rho_{21} \sigma_1,\delta_1 - \xi_{13}\delta_3\}}$ \\\hline
$O_{21}$ &  $\frac{-\max \{\sigma_1 - \rho_{12} \sigma_2,\delta_3 - \xi_{31}\delta_1\}}{\max \{\sigma_3 - \rho_{32} \sigma_2,\delta_2 - \xi_{21}\delta_1\}}$  \\\hline
$O_{22}$ &  $\frac{-\max \{\sigma_1 - \rho_{12} \sigma_2,\delta_1 - \xi_{12}\delta_2\}}{\max \{\sigma_3 - \rho_{32} \sigma_2,\delta_3 - \xi_{32}\delta_3\}}$ \\\hline
$O_{23}$ &  $\frac{-\max \{\sigma_1 - \rho_{12} \sigma_2,\delta_2 - \xi_{23}\delta_3\}}{\max \{\sigma_3 - \rho_{32} \sigma_2,\delta_1 - \xi_{13}\delta_3\}}$ \\\hline
$O_{31}$ & $\frac{-\max \{\sigma_2 - \rho_{23} \sigma_3,\delta_3 - \xi_{31}\delta_1\}}{\max \{\sigma_1 - \rho_{13} \sigma_3,\delta_2 - \xi_{21}\delta_1\}}$ \\\hline
$O_{32}$ &  $\frac{-\max \{\sigma_3 - \rho_{23} \sigma_3,\delta_1 - \xi_{12}\delta_2\}}{\max \{\sigma_1 - \rho_{13} \sigma_3,\delta_3 - \xi_{32}\delta_2\}}$ \\\hline
\end{tabular}
\caption{Saddle indices of the saddle equilibrium points.}
\label{table:indices}
\end{center}
\end{table}

\subsection{Persistence of $\Gamma_0$ and $T_0$ for small values of $|p|$}
Now we introduce in Eq.(\ref{4}) the weak non-zero coupling $p$ 
from the subsystem $X$ to the subsystem $Y$. 
\begin{enumerate}
\item[a).] \textit{Persistence of $\Gamma_0$}\\
To show the persistence of the heteroclinic network $\Gamma_0$ we consider
all heteroclinic orbits belonging to it. Without loss of generality we
choose the rectangle $\tilde{R}_0$; the proof for other rectangles is similar.

\underline{Persistence of $H(O_{33} \rightarrow O_{31}) = O_1 \times \tilde{\Gamma}_{31}$.}\\
This trajectory belongs to the three dimensional plane $x_2 = x_3 = 0$, $y_2 =
0$. Denote it by $\mathbb{R}^{3}_1$. This plane is invariant for $\varepsilon =
0$ both for the system (\ref{3},\ref{4}) and (\ref{syst_1},\ref{syst_2}).
Inside $\mathbb{R}^{3}_1$ the point $O_{33} = (\sigma_1, 0,\delta_3)$ is the
saddle equilibrium point for the system (\ref{3},\ref{4}) with eigenvalues 
$-\sigma_1$, $\delta_1 - \xi_{13}\delta_3$, $\delta_2 - \xi_{23}\delta_3$, see
\eqref{11}, i.e., with one-dimensional unstable manifold whereas
the point $O_{31} = (\sigma_1, 0,\delta_1)$ is the stable node with eigenvalues
$-\sigma_1$, $-\delta_1$, $\delta_3 - \xi_{31}\delta_1$, 
see Table~\ref{table:eigenvalues}.  Inside $\mathbb{R}^{3}_1$, for small values of $|p|$ 
the saddle (node) equilibrium point stays the saddle (node). 
Denote them by $O_{33}(p)$ and $O_{31}(p)$. The smooth
dependence of the unstable manifold on parameters and continuous dependence of
solutions of the ODE on parameters imply that for small values of $|p|$ there exists a
heteroclinic orbit joining $O_{33}(p)$ and $O_{31}(p)$. Persistence of
this heteroclinic trajectory is a structurally stable feature.

\underline{Persistence of $H(O_{33} \rightarrow O_{13}) = \Gamma_{31} \times
\tilde{O}_3$.}\\
This trajectory belongs to the three-dimensional invariant plane $x_2 = 0$, $y_1
= y_2 = 0$, say $\mathbb{R}^{3}_2$. Inside $\mathbb{R}^{3}_2$ the point $O_{33}
= (0, \sigma_3,\delta_3)$ is the saddle with the eigenvalues $\sigma_1 -
\rho_{13} \sigma_3$, $\sigma_2 - \rho_{23} \sigma_3$, $-\delta_3$,
and the point $O_{13} = (\sigma_1, 0,\delta_3)$ is the node with the eigenvalues
$- \sigma_1$, $\sigma_3 - \rho_{31} \sigma_1$, $-\delta_3$, 
The situation is structurally stable as well.

\underline{Persistence of $H(O_{31} \rightarrow O_{11}) = \Gamma_{31} \times
\tilde{O}_1$.}\\
This trajectory belongs to the three-dimensional invariant plane $x_2 = 0$, $y_2
= y_3 = 0$, say $\mathbb{R}^{3}_3$. Inside $\mathbb{R}^{3}_3$ (which is
invariant also for $p \neq 0$) the point
$O_{31} = (\sigma_3, 0,\delta_1)$ is the saddle with the eigenvalues $-
\sigma_3$, $\sigma_1 - \rho_{13} \sigma_3$, $-\delta_1$, 
and the point $O_{11}$ is the node with the eigenvalues 
$- \sigma_1$, $\sigma_3 - \rho_{31} \sigma_1$, $-\delta_1$. 
The heteroclinic trajectory joining $O_{31}$
and $O_{11}$ inside $\mathbb{R}^{3}_3$ is also structurally stable.

\underline{Persistence of $H(O_{13} \rightarrow O_{11}) = O_1 \times
\tilde{\Gamma}_{31}$.}\\
The heteroclinic trajectory belongs to the three-dimensional invariant plane 
$x_2 = x_3 = 0$, $y_2 = 0$, say $\mathbb{R}^{3}_4$. 
Inside it the point $O_{13} =(\sigma_1, 0,\delta_3)$ is the saddle 
with the eigenvalues $- \sigma_1$, $\delta_1 - \xi_{13}\delta_3$, $-\delta_3$, 
and  $O_{11}$ is the node, with the eigenvalues 
$-\sigma_1$, $-\delta_1$, $\delta_3 -\xi_{31}\delta_1$. 
Again, persistence of this trajectory is a structurally stable property.

Summarizing, we conclude:

\begin{theorem}\label{theorem_2}
Under the above conditions the heteroclinic network $\Gamma_0$ persists for
sufficiently small values of $|p|$.
\end{theorem}
\item[b).] \textit{Persistence of the heteroclinic attractor}

To show the existence of a heteroclinic attractor at weak  coupling $|p|$, 
we prove the persistence of all rectangles that form $T_0$. 
As an example, we take the rectangle $\tilde{R_0}$;
for other rectangles the proof is similar. The proof is based on the
following facts,

\begin{enumerate}
	
\item[(i).] All points $O_{ij}(p)$ belong to the invariant four-dimensional
space $\mathbb{R}^{4}_0$ $(x_2 = y_2 = 0)$. At small values of $|p|$,
they are saddle points in
$\mathbb{R}_4$ with two-dimensional unstable manifolds.
Heteroclinic trajectories between them that exist for small $|p|$
due to Theorem~\ref{theorem_2}, also belong to $\mathbb{R}^{4}_0$. 
Denote them by
$\Gamma_{ij}(p)$, so that $\Gamma_{ij}(0) = \Gamma_{ij}$, $i, j \in \{1,2,3\}$.

\item[(ii).] At $p= 0$, the point $O_{11}$ is the stable node in 
$\mathbb{R}^{4}_0$  
with negative eigenvalues  $- \sigma_1$, 
$-\delta_1$, $\sigma_3 - \rho_{31}\delta_1$, 
$\delta_3 - \xi_{31}\delta_1$, all of them disjoint from zero. 
Hence, at small values of $|p|$ it
is still a sink, and there exists an absorbing region $U$ with the
maximal attractor $O_{11}$ inside it: see Fig~\ref{fig:traj_rectangle}.
\begin{figure}[h]
\centerline{\includegraphics[width=0.25\textwidth]{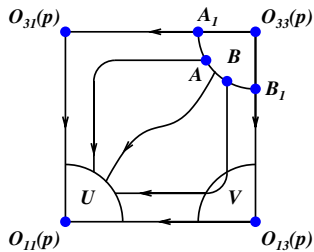}}
\caption{Trajectories on the rectangle.}
\label{fig:traj_rectangle}
\end{figure}
\item[(iii).] The local unstable manifold of $O_{33}(p)$ depends smoothly on
parameters, so for small values of $|p|$ it is $C^1$ close to
$W^{u}(O_{33}(0))$, the local unstable manifold for $p$=0. Therefore, if one
chooses an initial point $q$ on the interval $(A,B)$, 
Fig.~\ref{fig:traj_rectangle}, on $W^{u}(O_{33}(p))$ 
it will be close to a point $q_0 \in W^{u}(O_{33}(0))$. For the 
uncoupled system ($p = 0$) the trajectory passing through $q_0$ reaches $U$ 
in finite, bounded from above time. Thus, the trajectory of the system
(\ref{syst_1},\ref{syst_2}) passing through $q$ also
comes into $U$ in finite time if $|p|$ is small enough.

\item[(iv).] We show now that the representative point on the trajectory passing
through a point $\tilde{q} \in (A, A_1) \cup (B, B_1)$ comes eventually into
$U$. Without loss of generality we may assume that the point $B$ is so close to
$B_1$ that the trajectory passing through it intersects a small neighborhood $V$
of $O_{13}(p)$ at a point $\tilde{q}_1$ such that ${\rm dist}(\tilde{q}_1,
W^{u}(O_{13}(p)) <\delta$. We apply now the known results (see e.g., the
book~\cite{Sh}) to establish that 
${\rm dist}(\tilde{q}_2, W^{u}(O_{13}(p)) <\delta^{\nu}$
where $\tilde{q}_2$ is the point on the considered trajectory at the instant
when it leaves $V$ (see Fig.~\ref{fig:traj_in_V}) and $1 < \nu < \nu_{13}$, 
$\nu_{13}$ is the saddle index of $O_{13}(p)$.

\begin{figure}[h]
\centerline{\includegraphics[width=0.25\textwidth]{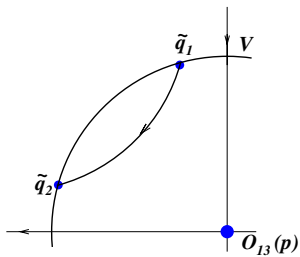}}
\caption{Trajectories in $V$.}
\label{fig:traj_in_V}
\end{figure}

This means that if $\delta$ is small enough, then the point $\tilde{q}_2$ 
is close to a point on the heteroclinic trajectory joining $O_{13}(p)$ 
and $O_{11}(p)$, therefore the trajectory  passing through $\tilde{q}_2$ 
reaches $U$ in finite time. 
The proof for points on $(A, A_1)$ is the same.

\item[(v).] Thus, we have proved that every trajectory passing though a point 
on $(A_1, B_1)$  enters $U$, and then tends to $O_{11}(p)$. The points on these
trajectories together with equilibrium points and heteroclinic trajectories
belonging to $\Gamma_0(p)$ form the desired rectangles $\tilde{R}_0(p)$.

The union of these rectangles forms the surface $T(p)$. The fact that $T(p)$ is
homeomorphic to $T_0$ follows directly from the construction above. It follows
\end{enumerate}
\end{enumerate}

\begin{theorem}\label{theorem3}
The attractor $T_0$ persists for small $|p|$.\\
\end{theorem}

Since the saddle indices of all equilibria depend continuously on p, 
the following statement holds:
\begin{theorem}\label{theorem4}
Under the conditions of Theorem~\ref{theorem3}, 
the torus $T(p)$ remains to be an attractor for small $|p|$.
\end{theorem}

\section{Numerical investigation}
\label{sect:numerics}
\subsection{General aspects}
For numerical studies of the system (\ref{syst_1},\ref{syst_2})
we fix the coefficients of linear terms at
$\sigma_1$=1, $\sigma_2$=1.1, $\sigma_3$=0.9, $\delta_1$=2.2,
$\delta_2$=2.1, and $\delta_3$=1.9.
For the matrices, the values
\begin{eqnarray}
\displaystyle
    \rho_{21}=0.6\,\frac{\sigma_{2}}{\sigma_{1}},\;&
      \rho_{31}=1.65\displaystyle\frac{\sigma_{3}}{\sigma_{1}},\;&
      \rho_{32}=0.7\displaystyle\frac{\sigma_{3}}{\sigma_{2}},\nonumber\\
      \rho_{12}=1.55\frac{\sigma_{1}}{\sigma_{2}},\;&
      \rho_{13}=0.62\frac{\sigma_{1}}{\sigma_{3}},\:&
      \rho_{23}=1.45\frac{\sigma_{2}}{\sigma_{3}}\nonumber
\end{eqnarray}
and
\begin{eqnarray}
    \xi_{21}=0.6\frac{\delta_{2}}{\delta_{1}},\; &
    \xi_{31}=1.65\displaystyle\frac{\delta_{3}}{\delta_{1}},\;&
    \xi_{32}=0.7\frac{\delta_{3}}{\delta_{2}},\,\nonumber\\
    \xi_{12}=1.55\frac{\delta_{1}}{\delta_{2}},\;&
    \xi_{13}=0.62\displaystyle\frac{\delta_{1}}{\delta_{3}},\;&
    \xi_{23}=1.45\frac{\delta_{2}}{\delta_{3}}\nonumber
\end{eqnarray}
are adopted. Finally, the coefficients $\eta_{ij}$ at mixed terms
obey $\eta_{ij}=i+0.2j^2$.       

In nine states of equilibrium of Eq.(\ref{syst_1},\ref{syst_2})
exactly two of three $x_i$ and  two of three $y_i$ vanish.
At the above values of coefficients, and at vanishing or sufficiently
weak coupling $p$, all these equilibria are saddles
with two-dimensional unstable manifolds. At $p$=0 the $x$- and $y$-subsystems
decouple; each of them possesses three saddles with one-dimensional unstable
manifold and the heteroclinic contour formed by the separatrices that
connect those saddles. The above choice of parameter values ensures 
(in terms of the corresponding saddle indices) that each contour 
is attracting in the partial subspace of the respective subsystem.  
In accordance with results of Sect.~\ref{sect:rigorous}, 
an attractor in the joint phase space at small values of $|p|$ 
should be a persistent two-dimensional torus $T_0$ 
with the heteroclinic network $\Gamma_0$ upon it 
(see Fig.~\ref{fig:torus_scheme}).

Peculiarities of dynamics near attracting heteroclinic contours
result in long epochs when a trajectory hovers in vicinities of saddle points.
Duration of these repetitive epochs grows exponentially, and a perfect
numerical integrator will, instead of delivering information about
the whole attracting state, exhaust the time resources in ever longer
passages near the unstable equilibrium. 
It is known that inevitable numerical 
inaccuracies (at least, at the roundoff level) and/or 
introduction of explicit noise  are able to 
kick the trajectories from the vicinities of the saddles; as a result
of these imperfections, a system with an attracting heteroclinic contour
displays virtually periodic behavior, with the ``period'' proportional to the
logarithm of the imperfection amplitude. Below, we introduce this
imperfection in the explicit controllable way; this should allow us to infer
the asymptotic properties of the unperturbed dynamics on $T_0$ 
from observable properties of perturbed numerical evolution.

Equations  (\ref{syst_1},\ref{syst_2}) possess invariant hyperplanes: $x_i$=0 
or $y_i$=0 $\forall i$. We impose impenetrable barriers parallel to these 
hyperplanes: none of the coordinates is allowed to vanish. 
In this way, we replace
the continuous dynamics  by a piecewise-continuous one:
after every timestep of integration of  (\ref{syst_1},\ref{syst_2}), 
the ``calibrations''
\begin{equation}
\label{epsilon}
x_i\to \max (x_i,\varepsilon),\;\;
y_i\to \max (y_i,\varepsilon),\;\;\;i=1,2,3
\end{equation}
are performed, with fixed small $\varepsilon>0$.
In this way, $\varepsilon$
becomes a governing parameter of the dynamical system.

\subsection{Measuring the instability rates}
A hallmark of a motion along the invariant two-dimensional surface are two
vanishing Lyapunov exponents. 
The top panel of Fig.~\ref{fig:lyap} shows in the decreasing order 
all six Lyapunov exponents, evaluated in the standard way 
(cf.\eqref{lyap_length}) for the trajectory of the system 
(\ref{syst_1},\ref{syst_2},\ref{epsilon}) at a relatively weak coupling
$p$=0.01 for $t=5\times10^5$ and $\varepsilon$ ranging 
from $10^{-3}$ to $10^{-36}$.
\begin{figure}
\centerline{\includegraphics[width=0.3\textwidth]{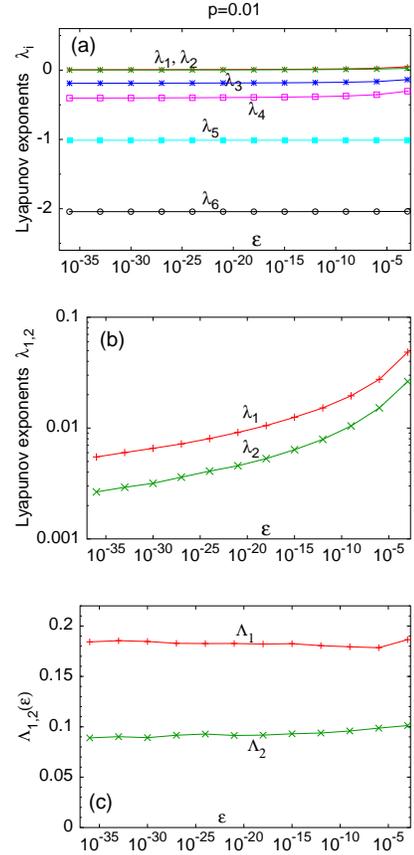}}
\caption{\label{fig:lyap}
Characteristics of instability for 
Eqs (\ref{syst_1},\ref{syst_2},\ref{epsilon}) at $p$=0.01.
(a),(b) conventional Lyapunov exponents; 
(c): length-related Lyapunov exponents}
\end{figure}
Indeed, at this level of graphical resolution we get an impression that 
at sufficiently small values of $\varepsilon$ 
the Lyapunov exponents tend to constant values, and their set with four 
negative and two vanishing values of $\lambda_i$ characterizes a motion 
along the attracting two-dimensional torus. 
However, a proper magnification of the region
adjacent to zero (middle panel of Fig.~\ref{fig:lyap}) discloses 
that the saturation for the two largest exponents is
deceptive: both exponents approach zero rather slowly, 
$\lambda_{1,2}(\varepsilon)\sim\ -1/\log \varepsilon$. 
At finite values of $\varepsilon$
the estimates $\lambda_{1,2}$ stay positive, indicating presence of two
modes of instability, albeit rather weak. 

Computation of conventional Lyapunov exponents $\lambda_i$ in accordance 
with (\ref{lyap_time}) is ambiguous for the two largest ones. 
We are much better served if, instead, we use the length-related
characteristics $\Lambda_i$, defined by Eq.(\ref{lyap_length})\footnote{For 
widespread situations with non-zero average speed of motion along the attractor 
in the phase space,  characteristics (\ref{lyap_time}) and (\ref{lyap_length})
are, up to a constant factor, equivalent. This does not hold for 
(\ref{syst_1},\ref{syst_2}) and similar systems where the average speed 
tends to zero in the limit $t\to\infty$.}.
As visualized in Fig.~\ref{fig:lyap}(c), estimates of both length-related 
exponents  $\Lambda_{1,2}$ stay nearly constant in the considered range 
of $\varepsilon$. This means that 
$\|\tilde{\mathbf{x}}_{1,2}(t)\| \sim \exp \big(\Lambda_{1,2}L(t)\big)$.

Similarly to the jargon based on conventional Lyapunov exponents where presence
of two and more positive LE is called ``hyperchaos'', here we can speak of
``weak hyperchaos''.
Of two positive exponents shown in Fig.~\ref{fig:lyap}c, the larger one
stems from the $y$-subsystem of the equations (\ref{syst_1},\ref{syst_2}) 
whereas the smaller one is related to the intrinsic dynamics 
of the $x$-subsystem. 

\subsection*{Choosing the appropriate length}
Before presenting reaction of $\Lambda_{1,2}$ to variation of the
coupling amplitude $p$, an important technical aspect should be discussed.
The flow  (\ref{syst_1},\ref{syst_2}) is a skew system: the variables $x_i$ act
upon the set $\{y_i\}$ without reverse action. By construction,
internal dynamics in the subspace $x_i$ ($i=1,2,3$) is independent of
the value of $p$. Three of the six Lyapunov exponents 
characterize evolution of perturbations 
within this subspace  and do not depend on $p$;
in Fig.~\ref{fig:lyap} these are, along with $\lambda_2$, 
the negative exponents $\lambda_3$ and $\lambda_5$.
Being restricted to the internal dynamics of the subsystem $x$,
the re-parameterized exponent $\Lambda_2$ should stay $p$-independent
as well. However, the \textit{total} length of the phase trajectory includes the
coordinates $y_i$ and thereby depends on $p$; hence, it cannot be used
in the evaluation of  $\Lambda_2$, and should be replaced there by the length
$L_x(t)$ of the projection onto the $x$-subspace. 
For a comparison of the growth rates of two instability modes,
we cannot express them in terms of \textit{different} lengths, hence 
below we substitute $L(t)$ in Eq.(\ref{lyap_length}) by $L_x(t)$ \textit{both}
for $\Lambda_1$ and $\Lambda_2$. 

Notably, for $\Lambda_1$ this procedure is not especially accurate 
at vanishing and very small values of $p$: as mentioned above, at $p=0$, 
$\Lambda_1$ characterizes the growth of instability in the \textit{isolated} 
subsystem $y$ where a normalization with respect to $L_y$ would be 
an obvious choice. 
At very small $p$, the influence of the coupling subsystem $x$ 
is weak, and an evaluation in terms of $L_x$ may distort the whole picture. 
This is confirmed in the left panel of  Fig.~\ref{fig_three_lengths}: 
there, the estimate of  $\Lambda_1$ based on $L_x$ approaches 
the horizontal asymptote at small 
$\varepsilon$ distinctly slower than analogous estimates based on
$L_y$ or the total length $L$. The length values have been determined 
in accordance with the following protocol: for all values of $\varepsilon$ 
(and further below, of the parameter $p$) the trajectory starts from the same 
initial conditions, and after a (discarded) transient of $10^3$ time units 
is further integrated for $5\times 10^5$ time units, producing a phase
curve in the 6-dimensional phase space. 
For this curve, we calculate its total Euclidean length $L$ 
as well as the lengths of its projections onto the three-dimensional 
$x$- and $y$-subspaces:  respectively, $L_x$ and $L_y$. 
\begin{figure}[h]
\centering{\includegraphics[width=0.49\textwidth]{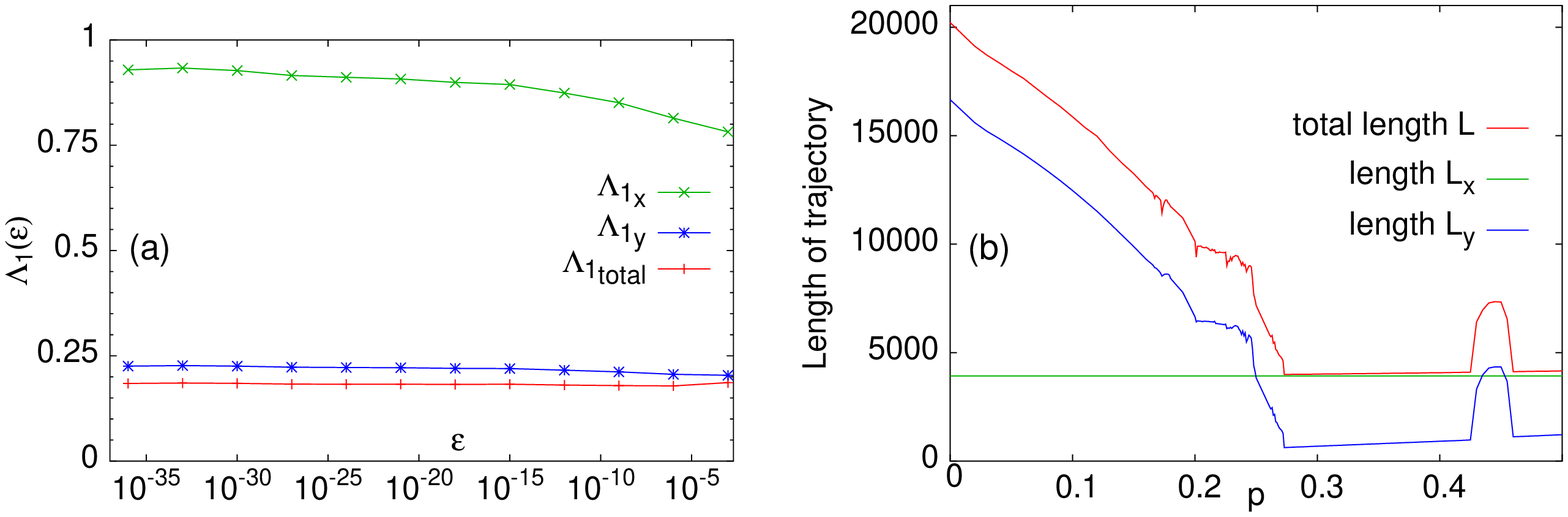}}
\caption{\label{fig_three_lengths}
Characteristic lengths for phase trajectories of  (\ref{syst_1},\ref{syst_2}).
Green, blue and red  curves correspond to measurements based on,
respectively, $L_x$, $L_y$ and the total length $L$. 
Integration interval: $t=5\times 10^5$.
(a): Instability rate (\ref{lyap_length}) at $p=0.01$ 
and variable $\bm\varepsilon$ in terms of different projections of trajectory.
(b): Dependence of lengths on the coupling strength $p$ 
at $\varepsilon$=$10^{-27}$.}
\end{figure}

\subsection{Variation of the coupling strength}

Variation of the parameter $p$ affects the dynamics of the system 
(\ref{syst_1},\ref{syst_2},\ref{epsilon}), 
both quantitatively and qualitatively.
In the expression for the growth rate (\ref{lyap_length}),
 this concerns $y$-related terms 
in the length $L$ of the reference orbit and in the norm
$\|\tilde{\mathbf{x}}(t)\|$ of the perturbation. 
We start with the influence of $p$ upon $L$
under fixed observation time $t$.
When, at constant $\varepsilon\neq 0$, the coupling $p$ is increased, 
repulsion near the saddles in the $y$-subspace weakens
(quantitatively, this can be read off the respective saddle indices), 
hence the system stays longer in the neighborhoods of those saddles, 
and the average speed of motion across the $y$-subspace lowers. 
As a result, the total length of the reference orbit, 
as well as the length of its $y$-projection, 
decrease. This effect is illustrated in the right panel of 
Fig.~\ref{fig_three_lengths}.  
At low values of $p$, $y$-related components dominate in the total length; 
at larger $p$, in contrast, 
$L_y$ becomes  much shorter than $p$-independent $L_x$, 
so that the difference  between the total  length and 
the length of $x$-projection becomes virtually negligible.
 
For estimates of the mean growth rates  $\Lambda_{1,2}$ 
of perturbations 
we use the $p$-independent length $L_x(t)$, so that the entire effect is 
due to changes in the value of  $\|\tilde{\mathbf{x}}(t)\|$. Recall that
the rate $\Lambda_2$ characterizes the internal dynamics in the $x$-subsystem
and is thereby insensitive to variations of $p$. In contrast,
the value of $\Lambda_1$ (determined at $p$=0 by dynamics in 
the subspace $y$), varies when $p$ is changed. 
This is illustrated in Fig.~\ref{fig_p_dependence}. 
\begin{figure}
\centering{\includegraphics[width=0.4\textwidth]{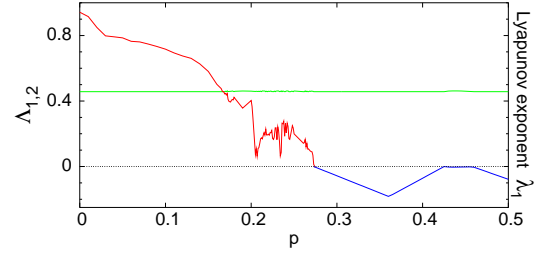}}
\caption{\label{fig_p_dependence}
Red and green solid curves, respectively: 
length-related exponents $\Lambda_{1}$ and $\Lambda_{2}$.
Blue solid curve: conventional (time-related) negative Lyapunov exponent.
The value of  $\varepsilon$ is fixed at $10^{-27}$.}
\end{figure}

As we see in the plot, increase of $p$ weakens this instability mode:
the value of $\Lambda_1$ nearly monotonically decays until, 
at $p\approx$0.175 it becomes smaller than the exponent 
$\Lambda_{2}$. Further growth of $p$ results in a jagged non-monotonic 
pattern: probably, an indicator of internal transitions in this weakly chaotic
state. Finally, $\Lambda_1$ changes sign close to $p=0.27$. 
There, this mode of slowly growing instability is replaced by the exponentially 
decaying perturbations, characterized by negative conventional Lyapunov 
exponent. Only one weakly chaotic component,
corresponding to evolution of $x$-variables persists; 
in the projections of  $y$-variables, there is almost no dynamics: 
practically all the time they hover close to the saddle points,
and the total length of the trajectory $L$ nearly coincides 
with the length $L_x$ of the $x$-projection.
In this way, the weakly chaotic dynamics close to the torus $T_0$ 
of the whole system is replaced by the ``simpler'' weakly chaotic dynamics 
near the heteroclinic contour of the master $x$-subsystem.
Around $p\approx 0.45$ there is a short range of $p$ where the negative
Lyapunov exponent nearly vanishes again; 
there, dynamics of the variables $y$ consists of short jerks,
and the length of projection onto subspace $y$ becomes comparable 
with $x$-projection (cf. Fig. \ref{fig_three_lengths}).

\subsection{Transformations in the phase space}
Changes in the instability rates, imposed by variation of $p$, are
reflected in changes of the phase portraits. In a reasonably broad
parameter range  $0\leq p\leq 0.27$, these changes  appear to
be mostly quantitative:
the shape of phase trajectories is qualitatively persistent.
Exemplary evolution of individual
variables at two values of $p$ from this range is presented 
in Fig.~\ref{fig_individual}. 
All variables display more or less ordered patterns. 
Recall that due to finite value of $\varepsilon$ 
the trajectories stay at a bounded distance 
from the invariant planes. 
For this reason, the times of residence in vicinities of the saddles, instead
of forming the growing geometric progression (as would be the case at 
$\varepsilon=0$), weakly oscillate near the constant values. Remarkably, these
values for two subsystems are different: there is no simple phase locking.
In both subsystems we observe the typical winnerless competition.
All three variables $x$ (top row) as well as all three $y$ variables 
(lower rows) oscillate in turn: while one of them traverses the high plateau, 
two other ones nearly vanish. 
For each variable in the $y$-system the plateaus consist of segments
with three different heights: 
$y_i=\delta_i-\,p\sum_j\eta_{ij}x_j$ where $x_j$ are the coordinates
of the master system at its saddle points. 
Differences between the plateau heights are proportional to $p$; 
they are weaker expressed in the left column of the plot, but are
well visible (especially for the variable $y_2$) in the right column. 
At $p=0.22$ the temporal pattern is apparently weakly disordered;
furthermore, the plateaus of $y_2$ are distinctly wider than the plateaus
of other driven variables: 
epochs of activity of $y_3$ and, especially, of
$y_1$ turn into the sharp spikes separated by uneven intervals.

\onecolumngrid

\begin{figure}[h]
\centering{\includegraphics[width=0.72\textwidth]{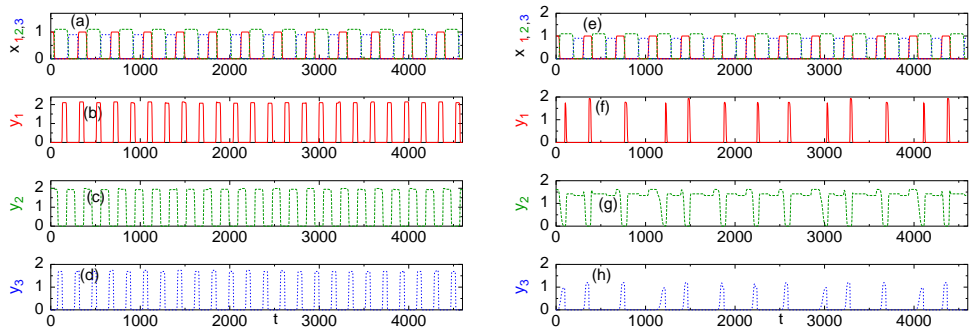}}
\caption{\label{fig_individual}
Dynamics on the non-smooth torus:
Temporal evolution of the system (\ref{syst_1},\ref{syst_2},\ref{epsilon})
at $\varepsilon=10^{-18}$. Left column: $p$=0.05; right column: $p$=0.22.
Top row: variables of the master system. Lower rows: variables of the
forced $y$ system.}
\end{figure}
\twocolumngrid

\begin{figure}
\centering{\includegraphics[width=0.47\textwidth]{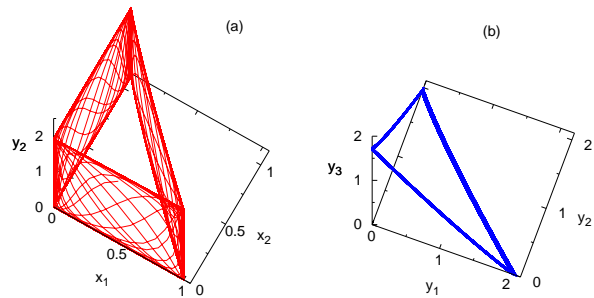}}
\caption{\label{fig_portraits}
Projections of the phase portrait of the system
(\ref{syst_1},\ref{syst_2},\ref{epsilon}) at $p=0.05$, $\varepsilon=10^{-18}$.}
\end{figure}

Characteristic projections of the phase portrait onto three-dimensional
subspaces for this oscillatory state are shown in Fig.~\ref{fig_portraits}. 
In the master $x$-subsystem (not shown) we observe just the attracting
heteroclinic contour. A plot with two coordinates from the master system
and one coordinate from the driven $y$-subsystem (left panel 
of Fig.~\ref{fig_portraits}) has the shape of a right triangular prism;
equilibria are projected onto vertices whereas the unstable manifolds
run along the edges. At finite $\varepsilon$ the attracting orbits escape 
the vertices along the faces of the prism.

The projection with coordinates entirely from the driven subsystem
(right panel of Fig. \ref{fig_portraits}) has a triangular shape.
Recall that each vertex of the triangle results from projecting of
three different (and well separated in terms of coordinates $x_i$) 
points of equilibrium from the $x$-system. 
In terms of the coordinates $y_i$, the distances between these three
projections are proportional to $p$, and a closer look in Fig.~\ref{fig_small}
shows that the triangle possesses a ``width'': its vertices 
(and, correspondingly the whole phase portrait) split into three components.
Deceptively close on the $y$-projection, these components are 
macroscopically separated in terms of the $x$-coordinates.

\begin{figure}
\centering{\includegraphics[width=0.3\textwidth]{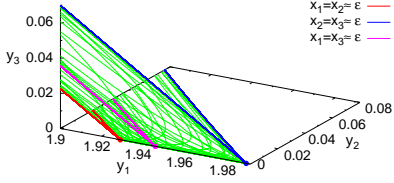}}
\caption{\label{fig_small}
Blowup of the phase space region near the $y_2$-axis. Filled circles 
and thicker curves: points of equilibrium and their separatrices.
Parameter values like in Fig.~\ref{fig_portraits}.}.
\end{figure}

At still higher values of $p$ the picture changes drastically: in the range
$0.28\leq p\leq 0.42$ only the variable $y_2$ survives in the driven system 
whereas  the variables $y_1$ and $y_3$ completely decay (numerically
both of them attain the value of $\varepsilon$).
An example of evolution of $y_2$ is presented in Fig.~\ref{fig_jerks}a:
it consists of horizontal plateaus connected by segments of rapid transitions. 
Comparison with dynamics of variables of the master system in the 
panel~\ref{fig_jerks}b
shows that each plateau corresponds to the epoch of activity for one of the 
master variables.
At 0.43$\leq p\leq$ 0.45  a further regime is observed 
(Fig.~\ref{fig_jerks}c,d) in which only $y_3$ decays 
whereas $y_1$ and $y_2$ alternate in activity. 
Finally, beyond $p=0.45$ the driven system returns to the state with 
only one active variable that jumps between three plateaus 
(Fig.~\ref{fig_jerks}e,f); 
this time it is the variable $y_1$ whereas $y_2$ and $y_3$ decay.
The reasons for this profound change in dynamics can be understood
from the plot of parameter dependence of the overall saddle indices
of the heteroclinic contours (see Sect.\ref{sect:rigorous}): 
the products of saddle indices 
over all saddles participating in the contour~\cite{Afr_MIR_Varona_2004}. 

%\newpage

\onecolumngrid

\begin{figure}[h]
\centering{\includegraphics[width=0.9\textwidth]{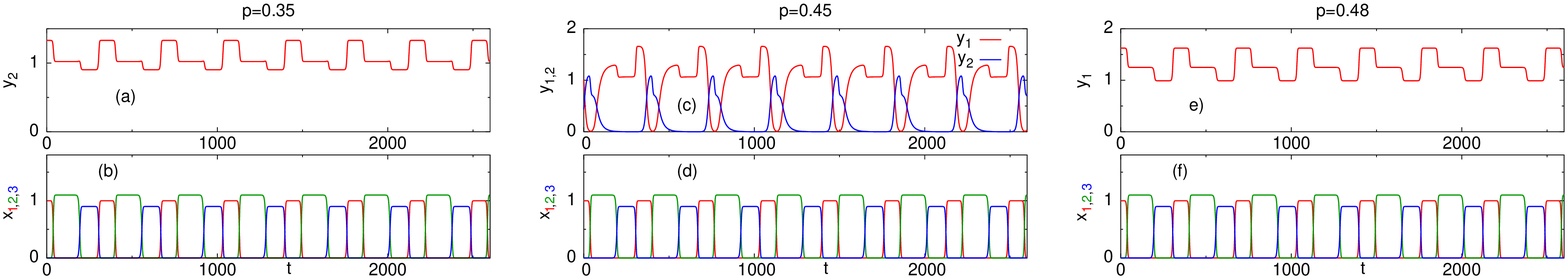}}
\caption{\label{fig_jerks}
Dynamics after the breakup of the non-smooth torus:
Temporal evolution of the system (\ref{syst_1},\ref{syst_2},\ref{epsilon})
at $\varepsilon=10^{-18}$.\protect\  
Left column: $p$=0.35; middle column: $p$=0.45; 
right column: $p$=0.48.\protect\ 
Top row: non-decaying variables of the driven system:
(a) $y_2$;  (c) $y_{1,2}$; (e)  $y_1$.\protect\  
Bottom row: variables of the master system: $x_1,x_2,x_3$.}
\end{figure}
%\hrule

\twocolumngrid

The non-smooth torus at small $|p|$ is built from several heteroclinic 
contours: when the master $x$ is frozen at one of its saddle equilibria,
the driven system has three saddle points 
whose one-dimensional unstable manifolds form a contour. 
Altogether there are three such contours, and
for each of them the overall saddle index should be checked separately, 
taking into account only the eigenvalues pertaining to the $y$-subspace.
Recall: a heteroclinic contour is attracting if its overall index exceeds 1.

Fig.~\ref{fig_indices} shows dependence of all three overall indices 
on the coupling strength $p$. For each curve, an initial weak decrease 
is superseded by subsequent growth: attraction to the contours becomes 
stronger. Remarkably, the saddle indices of \textit{separate} saddles may 
decrease and fall below 1 (not shown in Fig.~\ref{fig_indices}); 
what matters, however,  are not the separate indices 
but their overall product -- and it grows!
Furthermore, one by one, the values of the overall products diverge: 
one of the respective positive eigenvalues 
becomes small and finally vanishes:
the equilibrium loses instability 
and turns from the saddle into the stable node.

\begin{figure}
\centering{\includegraphics[width=0.35\textwidth]{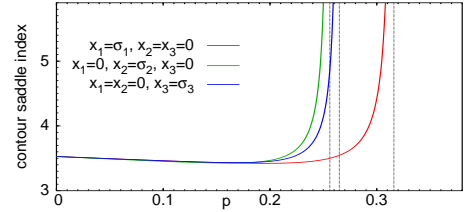}}
\caption{\label{fig_indices}
Overall saddle indices, composed of eigenvalues, leading
in the $y$-subspace.
Coloring denotes location of the %corresponding 
saddle point in the $x$-subspace:
$(x_1$=$\sigma_1,\,x_2$=$x_3$=0)  for the red curve,
$(x_1$=0, $x_2$=$\sigma_2,\,x_3$=0)  for the green curve, and
$(x_1$=$x_2$=0, $x_3$=$\sigma_3)$  for the blue curve. 
Dotted vertical lines:
stabilization (via transcritical bifurcations) of respective saddles.
Note: growth of \textit{overall} indices does not exclude the decrease
of indices of \textit{separate} saddles (not shown here).}
\end{figure}

For an equilibrium lying on the coordinate axis of the $y$-subsystem,
stabilization is a result of the transcritical bifurcation. 
With the master system frozen at one of its saddles, 
the driven subsystem, along with a set of 
three such ``axial'' equilibria,  possesses three further steady states, 
one in each of the invariant planes $y_i$=0, $i=1,2,3$.
At $p$=0 all three ``in-plane'' states have, beside zero, 
a positive and a negative coordinate $y_i$; thereby, they lie in the 
``non-physical'' part of the phase space. 
Two of these states are asymptotically stable in the subspace $y$ whereas 
the third one is a saddle with one negative and two positive Jacobian 
eigenvalues. Since the coordinates of these equilibria are linear
functions of the coupling $p$, some
of them move into the positive octant when $p$ is increased. 
Crossing on their way the pertaining
coordinate axis, they collide with the corresponding ``axial'' equilibrium.
In the course of this transcritical bifurcation, the equilibria exchange 
stability, and the axial one becomes stable. 
Explicit expressions for the bifurcation parameter values are too
lengthy to be quoted here.

Stabilization of the former saddle(s) changes the dynamics of $y$ 
at frozen $x$: now there is a simple attractor, the subsystem 
eventually converges to it and stays there forever. 
``Unfreezing'' $x$ leads to rapid switches in the dynamics of 
the driven subsystem: arrival of the master system at its next saddle point 
implies for the driven one a new structure of the phase space,
where there may still be three saddle points and a heteroclinic contour, 
or, instead, a steady ``axial'' attractor at which the driven system resides 
until the next switch. In the latter case two of the coordinates $y_i$ decay 
whereas the third one assumes the equilibrium value. 
In the course of the further increase of $p$, 
stabilization of saddle states, step by step,  occurs in all three ``frozen'' 
subspaces, and finally heteroclinic dynamics in the driven system dies out.
In the situations of the left and right columns of Fig.\ref{fig_jerks},
the stable equilibria before/after the switch lie on the same coordinate axis, 
hence the evolution of the driven subsystem becomes 
effectively one-dimensional; in the middle column,  the driven
subsystem jumps between equilibria on two axes, 
keeping the third coordinate negligible.

The states in the left and right columns of Fig.\ref{fig_jerks} illustrate
replacement of the winnerless competition in the slave subsystem $y$ 
by the quasi-steady ``winner-take-all'' situation: 
As long as the master $x$ remains in the
nearly static configuration at one of its saddle-points, the
slave subsystem synchronizes with it, becoming static as well.
Nontrivial dynamics emerges in the situation of ``nimble master, 
lazy slave'' when the typical time of heteroclinic switching in the master 
subsystem is smaller than the time of relaxation
to the equilibrium in the slave subsystem; details of this dynamics
will be reported elsewhere.

Alteration of quasi-steady states also explains the practically 
piecewise-linear 
dependence on $p$ of the negative Lyapunov exponent (blue solid curve in 
Fig.~\ref{fig_p_dependence} above). In those states, the driven subsystem jumps 
between the (emerging) stable equilibria, and the Lyapunov exponent is just the
weighted sum of the least negative Jacobian eigenvalues of these equilibria;
the weights are the normalized lengths of the corresponding plateaus.
Since all eigenvalues are linear functions of $p$, their weighted sum is it 
as well.

\section{Discussion}

In this paper, we have discussed coordination among
coupled heteroclinic networks, whose dynamics
mimics sequential switching of metastable information units.
Coupled networks of this kind exist 
on different levels of brain elements hierarchy.
The hierarchy itself results from  complex functional interactions, 
residing between the poles of segregation and integration 
tendencies for networks that perform joint  
specific cognitive and/or behavioral tasks.
In particular, we have suggested the dynamical mechanism 
of low-dimensional coordination that is
related to the general information processing in the brain 
sequential units~\cite{Baldassano_et_al_2017}.
The key phenomenon of this coordination is 
entrainment of localized units  in multimodal brain activity~\cite{RAV_2010}. 

For the upper level of network hierarchy 
we have proposed a low-dimensional mathematical model 
of the brain-to brain interaction. 
The model belongs to the class of generalized Lotka-Volterra systems
that, when decoupled, feature rhythmic activity.
Regimes, observed in the phase space 
in the case of the simplest master-slave configuration, 
can be viewed as mathematical images of the corresponding cognitive processes. 
Under sufficiently weak coupling, the attractor is a non-smooth 
two-dimensional torus that contains equilibrium points 
of the saddle type and is composed of heteroclinic orbits joining those points. 
Instability of all trajectories in the basin of the attractor is confirmed 
by presence  of two positive length-related Lyapunov exponents. 
Typical pattern of behavior  on the attractor is successive switching 
between the saddles  along different heteroclinic channels~\cite{AGR}. 
Under stronger coupling the system undergoes a bifurcation 
related to the breakup of the invariant torus. 
After the breakup, the winnerless competition dynamics
in the slave subsystem is replaced by the winner-take-all quasistatics.
This restricts the number of options in the slave subsystem and makes
cooperation between the master and the slave more rigid. 

\textit{Proper account of fluctuations.} \  
In numerical studies we have substituted the fluctuating terms in the equations
(e.g., explicit multiplicative noise) by the formal construction that keeps
trajectories from coming too close to the invariant hyperplanes. 
Replacement is justified by the fact that the sole role of fluctuations, 
regardless of their explicit shape, is to kick the system 
out of the vicinities of the equilibria
where, otherwise, the system would spend the overwhelming 
proportion of its time.
Therefore, we expect qualitatively the same results if, instead, the
stochastic version of Eq.\eqref{LV_two} is simulated.

\textit{A few words about bidirectional coupling.} \ 
Our analysis has been restricted to unilateral coupling;
in the dynamical system \eqref{LV_two} this corresponds to vanishing 
parameter $q$, responsible for the reverse influence of the participant $Y$ 
upon the participant $X$. Part of our results can be extended to the case 
of bidirectional interaction; 
this refers, in particular, to the existence at weak coupling rates
of the attracting non-smooth torus  with the heteroclinic network. 
The proof in  Sect.~\ref{sect:rigorous} is based 
on presence of the torus in the case of decoupled subsystems 
and on the continuity arguments for sufficiently weak unilateral coupling $p$. 
Similar continuity arguments ensure persistence of this attractor 
under sufficiently small values of reverse coupling $q$ as well. 
Accordingly, the system with weak two-way coupling should also feature the 
``toroidal'' winnerless competition, with heteroclinic channels 
between the saddles  formed not along one-dimensional separatrices 
but along two-dimensional manifolds. Substantial increase of either (or both) 
of the coupling coefficients $p$ and $q$ enforces the breakup 
of the non-smooth torus; along with the mechanism described above 
(partial regain of stability by the saddle points  in one of the subsystems), 
other scenarios can develop as well, 
e.g., loss of attraction by some of the heteroclinic contours 
and subsequent ``smoothing'' of respective corners of the attractor. 
Details of these effects will be reported elsewhere.

\textit{Master-slave case: are we slaves of our memories?}
As mentioned in the Introduction, unidirectional configuration of
coupling can model the influence of episodic memory in the past
upon memory dynamics in the future. In this respect, the winner-takes-all
behavior described in the end of the preceding section might be 
of interest for certain kinds of mental disorder where attention of a patient 
is rigidly fixed at a few past events: the past memory (master) 
cyclically switches between several episodes with long stay at each of them; 
during these stays, the current memory  (driven subsystem) stays frozen, 
but as soon as the past episode changes, the equilibrium of the current memory 
ceases to exist,
and the memory abruptly moves on to its new tentative attractor.

\textit{Brain-to-Brain information generation.}
Episodic memory for real life
involves the  orchestration of multiple time scales dynamical processes,
including  hierarchical chunking  and multimodal binding of events.  To
concentrate at the core of the phenomenon of episodic entrainment, 
we have restricted our treatment to the simplest approximation.
We supposed that the characteristic time of episodes forming i.e., chunking 
$t_{\rm ch}$  is much shorter than the characteristic time $t_{\rm ep}$
of sequential switching between episodes. 
Based on the generalized hierarchical model of episodic 
memory~\cite{Varona_Rabinovich_2016}, 
entrainment with arbitrary ratio 
$t_{\rm ch}\ll t_{\rm ep}$ can be considered. 
Interesting new dynamics is expected within this modeling framework. 
In particular, confusion or entanglement of the events from different 
episodes in the entrainment memory can occur. 
The dynamical origin of such memory errors can be the overlap of weakly 
chaotic time series representing different episodes in the sequential 
entrainment process (about the neurophysiological origin of the errors 
and distortion in the episodic memory see, for  example, \cite{Schacter_2012}). 
To estimate the level of information generated in the error sequence 
of episodes, the technics suggested in~\cite{Rabinovich_etal_2012} 
can be employed.

\begin{acknowledgments}
Research in Sect.\ref{sect:rigorous} was carried out with the financial
support for V.A. of the Russian Science Foundation (Project No 16-42-0143).
Numerical studies of M.Z.  in Sect.\ref{sect:numerics}
were supported by the RSF  (Project No. 17-12-01534).
\end{acknowledgments}

\end{document}